\documentclass[12pt]{article}
\usepackage[scaled]{helvet}
\usepackage[nodisplayskipstretch]{setspace}

\usepackage[utf8]{inputenc}
\usepackage{amsmath, amsthm, bm}
\newtheorem{assumption}{Assumption}
\usepackage{amsfonts,comment}
\usepackage{multicol, booktabs, bbm, color, xcolor}
\usepackage{natbib}
\usepackage{graphicx,tcolorbox,xcolor}
\usepackage{subcaption}
\usepackage{lscape}
\graphicspath{{./img/}}

\usepackage{mathtools}

\DeclarePairedDelimiter\floor{\lfloor}{\rfloor}

\newcommand{\oA}{\overline{A}}
\newcommand{\oL}{\overline{L}}
\newcommand{\indep}{\perp \!\!\! \perp}

\usepackage{titlesec}
\newtheoremstyle{asu}
  {2pt}
  {}
  {}
  {}
  {\bfseries}
  {.}
  {.5em}
  {}
\newtheoremstyle{thm}
  {}
  {2pt}
  {\itshape}
  {}
  {\bfseries}
  {.}
  {.5em}
  {}
\newtheoremstyle{remarks}
  {3pt}
  {4pt}
  {}
  {}
  {\bfseries}
  {.}
  {.5em}
  {}

\newtheorem{theorem}{Theorem}[section]

\newtheorem{lemma}[theorem]{Lemma}

\theoremstyle{asu}
\theoremstyle{definition}

\newtheorem{remark}[theorem]{Remark}

\begin{document}
\def\spacingset#1{\renewcommand{\baselinestretch}%
{#1}\small\normalsize} \spacingset{1}

\title{A structural nested rate model for estimating the effects of time-varying  exposure on recurrent event outcomes in the presence of death}
\date{}
\author{Daniel Mork\\
Department of Biostatistics\\ 
Harvard T.H. Chan School of Public Health\\\\
Robert L. Strawderman\\
Department of Biostatistics and Computational Biology\\
University of Rochester\\\\
Michelle Audirac\\
Department of Biostatistics\\ 
Harvard T.H. Chan School of Public Health\\\\
Francesca Dominici\\
Department of Biostatistics\\ 
Harvard T.H. Chan School of Public Health\\\\
Ashkan Ertefaie\\
Department of Biostatistics, Epidemiology and Informatics\\
University of Pennsylvania}
\maketitle

\newpage
\begin{abstract}
Assessing the causal effect of time-varying exposures on recurrent event processes is challenging in the presence of a terminating event. Our objective is to estimate both the short-term and delayed marginal causal effects of exposures on recurrent events while addressing the bias of a potentially correlated terminal event. Existing estimators based on marginal structural models and proportional rate models are unsuitable for estimating delayed marginal causal effects for many reasons, and furthermore, they do not account for competing risks associated with a terminating event. To address these limitations, we propose a class of semiparametric structural nested recurrent event models and two estimators of short-term and delayed marginal causal effects of exposures. We establish the asymptotic linearity of these two estimators under regularity conditions through the novel use of modern empirical process and semiparametric efficiency theory. We examine the performance of these estimators via simulation and provide an R package \textit{sncure} to apply our methods in real data scenarios. Finally, we present the utility of our methods in the context of a large epidemiological study of 299,661 Medicare beneficiaries, where we estimate the effects of fine particulate matter air pollution on recurrent hospitalizations for cardiovascular disease.
\end{abstract}

\noindent%
\vfill

\newpage
\spacingset{1.9} 
\section{Introduction}
\addtolength{\textheight}{.5in}%

The detrimental health effects due to long-term exposure to air pollution are well documented \citep{WorldHeathOrganization2016AmbientDisease}. Fine particulate matter (PM$_{2.5}$), in particular, has been shown to increase the risk of cardiovascular disease (CVD) and mortality \citep{Brook2010ParticulateDisease,Correia2013The2007,Di2017AirPopulation,Miller2007Long-TermWomen}. Most epidemiological research on the adverse effects of air pollution has considered the time to the first hospitalization or the time to death as a result \citep{Gerard2001TheEpidemiology, Shah2013GlobalMeta-analysis,Zhang2022AirInfarction}. 
However, to fully understand and quantify the burden of air pollution health effects, it is essential to consider the recurrent aspect of non-terminal health outcomes, as well as the time-varying nature of exposure and potential confounders.

Despite a rich body of work detailing methods for the analysis of recurrent event data, limited attention has been given to developing methods for assessing the causal effect of exposures on a nonterminal health outcomes (e.g. hospitalization), in the presence of a terminating event (e.g. death). Inverse probability weighting-based approaches have been proposed to adjust for death- and event-dependent censoring \citep{RLS2000,cook2009robust, schaubel2010estimating}, but these methods are limited to a single time point measure of exposure and rely on strong independent censoring assumptions.  \cite{Jensen2016AHospitalisations} proposed a marginal structural model to adjust for time-dependent confounders when studying the causal effect of a time-varying exposure on a recurrent event. However, inference for the corresponding estimators is based on correctly specified parametric models for the weight functions, a condition unlikely to be satisfied. Additionally, these methods may exhibit instability when there is a long trajectory of exposures. \cite{Su2020DoublyData} proposed a doubly robust semiparametric estimator multiplicative rate model for recurrent event data; see also \citet{Su2022CausalPseudo-observations}.  
However, their method does not account for time-varying exposures, confounding, or terminal events. Several other recurrent event methods adjust for correlated terminal events, but do not consider exposures that vary over time \citep{Schaubel2006AData, Liu2004SharedEvent,Ye2007SemiparametricEvents,Sun2012AnEvent,Yu2013ACoefficients,Chen2012SemiparametricEvents}. Proportional rate models provide an alternative approach  that allow for inclusion of external time-varying covariates (i.e., covariates that can be observed even after death); see, for example, \citet{ghosh2002marginal,Miloslavsky2004RecurrentCensoring, Sun2009RegressionEffects} and \citet{Li2016RecurrentCovariates}. Under these approaches, the target estimate is the direct effect of time-varying exposures on the event rate, without accounting for the potential indirect effect that is mediated through future time-varying covariates.  Hence, these methods fail to provide an unbiased estimate of the marginal causal effects of exposures, something that is often of interest to policy makers.

 Structural models have been used in survival analysis to estimate the causal effect of a given exposure on failure times \citep{robins1994adjusting,hernan2005structural, picciotto2016g}. \cite{martinussen2011estimation} developed a method to estimate a causal effect in the context of a point exposure assuming an additive hazard scale;
see also \cite{ dukes2019doubly}. \cite{seaman2020adjusting} proposed a class of semiparametric structural nested cumulative survival time models that can handle exposures varying over time, extending the work of \cite{picciotto2012structural} to continuous time. The main advantage of structural nested models is their ability to estimate the immediate and lagged causal effects of exposures on a failure time. However, these methods are not applicable to recurrent event settings.


In this work, we propose a new statistical learning framework to estimate  the short-term and delayed causal effects of time-varying exposures on recurrent events while accounting for bias arising from a potentially correlated terminal event. Specifically, we introduce a novel class of semiparametric Structural Nested Cumulative Number of Recurrent Events models (SNCURE) and propose two estimators for these causal effects. Both estimators are grounded in structural models for the potential recurrent event rate. Our first estimator relies on a parametric assumption for the exposure model but does not require parameterizing the effect of time-varying confounders for the recurrent event rate (i.e., the outcome model). Our second estimator relaxes the parametric assumption in the exposure model but requires modeling the effect of time-varying confounders on the recurrent event rate model. Compared to the first estimator, the second estimator is more robust, as it allows for nonparametric estimation of nuisance parameters while still enabling root-n inference under suitable regularity and rate conditions. 

Our work makes several novel contributions. First, to the best of our knowledge, we propose the first approach to estimate the short-term and delayed causal effects of time varying exposure on recurrent events while accounting for bias arising from a potentially correlated terminal event. The proposed SNCURE generalizes the structural model of \cite{seaman2020adjusting} to recurrent event settings by explicitly addressing bias due to competing risks of a potentially correlated terminal event.
Second, in addition to an estimator where nuisance functions are modeled parametrically, we introduce an estimation procedure that allows nuisance functions to be estimated nonparametrically. Third, we characterize the corresponding influence functions and identify the necessary conditions for the asymptotic linearity of the estimators. Notably, \cite{seaman2020adjusting} do not derive influence functions for their estimators. As such, our results provide a theoretical foundation that extends to their setting as a special case, where the event process corresponds to a terminal event rather than a recurrent one.

Our paper proceeds as follows. In Sections 2 and 3 we introduce notation and the structural nested modeling approach. Section 4 details our two estimators, including asymptotic results. We provide empirical evidence supporting the asymptotic results in a comprehensive data simulation in Section 5. Finally, in Section 6 apply our method to estimate the effects of fine particulate matter air pollution on hospitalization for recurrent cardiovascular disease among a sample of 299,661 Medicare beneficiaries residing in the Northeast region of the United States.

\section{Notation and Assumptions}
Figure \ref{fig:study_design} visualizes the study design in the context of our data illustration. We consider $n$ individuals indexed by $i$ and observed in continuous time, $t$. Each individual is followed for a baseline period that spans $M$ months and then during a study period that begins at time $t=0$ until $t=\tau$. We stop following an individual early (i.e., before time $\tau$) if they die or leave our study cohort. The time of death or censoring is indicated by $D$ and $C$, respectively, where $C=\tau$ indicates administrative censoring. We use $k=\floor{t}$ as a discrete index (e.g., representing whole days, months, or years since $t=0$) of the time-varying data where $\floor{x}$ is the floor function and $K=\floor{\tau}$. 
We note that although $k$ is an integer, $t$ is a real number to identify the precise moment of the event or an exposure measurement is taken. We omit the index $i$ from the notation unless needed for clarity.

\begin{figure}[!ht]
    \centering
    \includegraphics[width=\textwidth]{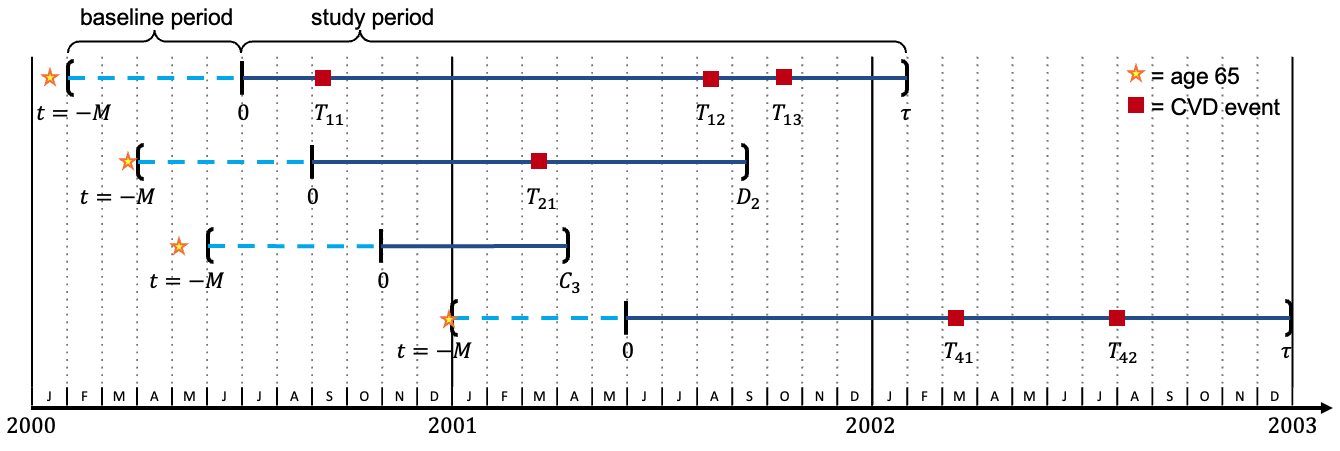}
    \caption{Diagram of the study design with $n=4$ individuals in the context of our data illustration representing a dynamic cohort of Medicare beneficiaries and recurrent cardiovascular disease hospitalizations. Each individual is followed during a baseline period spanning $M$ months after they turn age 65 (indicated by $\star$) and then during the study period beginning at $t=0$ until end of study ($\tau$), death denoted $D_i$ for individual $i$, or censoring by leaving the cohort for other reasons denoted $C_i$. Recurrent events occur at times $0<T_{i1}<\cdots <T_{i N_i(\tau)}<\tau$.
    }
    \label{fig:study_design}
\end{figure}

Let $A_{k}$ denote the average exposure  during times $t\in[k,k+1)$ while $L_{k}$ represents a vector of covariates during the same time period. For example, individual 1 in Figure \ref{fig:study_design} has measured exposure $A_{0}$ and covariates $L_{0}$ in July 2000; for individual 2 $A_{0}$ and $L_{0}$ are measured in September 2000. Next, let $\oA_{k}=[A_{-M},\ldots,A_{k}]$ denote the exposure observation history before and including month $k$ with $\oL_{k}$ representing the corresponding history of covariates. For example, $\oA_{-1}=[A_{-M},\ldots,A_{-1}]$ would indicate all exposures during the baseline period, while $\oA_{\floor{\tau}}$ is the history of exposures throughout the baseline plus study periods. For any $k< -M$, $\oA_{k}=\oL_{k}=\emptyset$.

Define $Y(t)=\mathbb{I}(X \geq t)$ to be a left continuous at risk process, where 
$X = D\wedge C$ is an indicator of alive and uncensored with $a\wedge b=\min(a,b)$. For the moment, we assume $C=\tau$ for all individuals, which requires all subjects to have the same amount of follow-up time (i.e., provided that death does not occur before $\tau$). In Section \ref{sec:censor} we consider non-administrative censoring. Define $N(t)$ as the observed recurrent event counting process for a given individual at time $t\in [k,k+1)$; let the corresponding underlying conditional rate function be $E\{dN(t)|\oA_{k},\oL_{k},D\geq t\}$.  We assume that recurrent events are not observed after $X$ (i.e., failure or censoring) so that $N(t) = N(X)$ for any $t \geq X$. Denote the observed failure process $\tilde N(t) = I(X \leq t, \Gamma=1)$ as the usual counting process for the terminal event where $\Gamma=I(D \leq C)$. In our formulation, we assume that exposure $A_k$ (e.g., PM$_{2.5}$ level) is an external variable that is available regardless of the occurrence of a censoring or terminating event. However, for covariates and exposures that vary over time, we define $L_{k}=\emptyset$ for any $k>X$. We assume that our data consists of $n$ independent, identically distributed trajectories $\{L_{-k}, A_{-k}, N(-k),\tilde N(-k)\}$, for $k=-M,-M+1,\cdots,K$. 

Define $N_{(\oA_{k-m},0)}(t)$ and $\tilde{N}_{(\oA_{k-m},0)}(t)$ for $t\in[k,k+1)$ as counterfactual counting processes when $\oA_{k-m}$ is fixed at the observed value and other exposures $A_{k-m+1},\cdots, A_{k}$ are set to zero. 
Under counterfactual notation, the case where $m=0$ is equal to the observed counting process, that is, $N_{(\oA_{k},0)}(t)= N(t)$ and $\tilde N_{(\oA_{k},0)}(t)= \tilde N(t)$ for $t\in[k,k+1)$. We note that $N(t)$ and $\tilde{N}(t)$ are unchanged when intervening in future exposures. In other words, $N_{(\oA_{k-1},0)}(k)= N_{(\oA_{k},0)}(k)=\cdots =N_{(\oA_{K},0)}(k)$, and $\tilde{N}_{(\oA_{k-1},0)}(k)= \tilde{N}_{(\oA_{k},0)}(k)=\cdots =\tilde{N}_{(\oA_{K},0)}(k)$ for terminal events. Finally, let $D(\oA_{k-m},0)$ denote the counterfactual terminal event time when $\oA_{k-m}$ is fixed at the observed value and the other exposures are set to zero.

Let $\mu_{km}(t)=\mathbb{E}[A_{k-m}|\oA_{k-m-1},\oL_{k-m},D(\oA_{k-m},0) \geq t]$ denote the time-varying exposure mechanism. The following assumptions will link the counterfactual and observed data.

\begin{assumption} (Causal assumptions) \label{assump:causal} 
\begin{enumerate}
    \item[]a. Consistency:   For $t\in[k,k+1)$ $N(t)=N_{(\bar a_{k},0)}(t)$ and $\tilde N(t)=\tilde{N}_{(\bar a_{k},0)}(t)$ if $\oA_{k} = \bar a_{k}$; 
    \item[]b. No unmeasured confounding (NUC): For $t\in[k,k+1)$ there is no unmeasured confounding between exposure $A_{k}$ and  counterfactual counting process $N_{(\oA_{k-1},0)}(t)$ given observed exposure and covariate history, i.e.,
    $N_{(\oA_{k-1},0)}(t)\indep A_{k}|\oA_{k-1}, \oL_{k},D\geq k.$ 
    This also applies to the counterfactual terminal event counting process,
    $\tilde{N}_{(\oA_{k-1},0)}(t) \indep \linebreak A_{k}|\oA_{k-1}, \oL_{k},D\geq k.$ Finally, we also assume $D(\bar A_{k-1},0) \indep A_{k}|\oA_{k-1}, \oL_{k},D\geq k$.
    \item[]c. Positivity: for an exposure $a\in\mathbb{A}_k$ and all possible covariates $\{t, \bar a_{k-1},\bar l_k\}$, $t\in[k,k+1)$, the conditional density $f_{A_{ik}| \oA_{i(k-1)},\oL_{ik}} (a| t, \bar a_{k-1},\bar l_k) > \epsilon$ for some $\epsilon > 0$.
\end{enumerate}
\end{assumption}

\noindent Assumption \ref{assump:causal}b allows us to isolate the exposures from the counterfactual recurrent event and the terminal event process, in particular, for an individual still alive at time $k$, exposure $A_k$ depends only on $\oL_k$ and $\oA_{k-1}$ and not on future recurrent events or the remaining lifetime.

\section{Structural nested rate model}

Consider the following structural model for the counterfactual counting process of interest 
\begin{multline}\label{eq:struct_rate1}
    \mathbb{E}\left[
        dN_{(\oA_{k-m-1},0)}(t)| \oA_{k-m}, \oL_{k-m}, D(\oA_{k-m-1},0)\geq t
    \right]\\
    =\mathbb{E}\left[
        dN_{(\oA_{k-m},0)}(t)|\oA_{k-m}, \oL_{k-m}, D(\oA_{k-m},0)\geq t
    \right]-
    A_{k-m} \tilde \beta_m dt,   
\end{multline}
for $t \in [k,k+1)$ with $m=0,\cdots,M$. In this formulation, the target parameter
$\tilde \beta_m$ respectively captures
the short-term ($m=0$) and delayed ($m>0$) effects of exposure $A_{k}$ on the event rate; under Assumption \ref{assump:causal}b, the term $\mathbb{E}\left[
        dN_{(\oA_{k-m-1},0)}(t)| \oA_{k-m}, \oL_{k-m}, D(\oA_{k-m-1},0)\geq t
    \right]$ is independent of $A_{k-m}$, and thus,  can be interpreted as a baseline rate function. 
We note that this model is similar to the semiparametric additive hazards model described by \citet{Dukes2019OnDifference} with additional lagged exposure effects. 

We begin by describing a structural model to isolate the short-term effects of exposure (i.e., $\tilde \beta_0$) on the recurrent event rate (i.e., $dN_{(\oA_k,0)}(t)$ for $t\in[k,k+1)$). Consider setting the most recent exposure
$A_k=0;$ 
then, the expected change in event rate due to lowering the most recent exposure can be described through the relation
\begin{equation}\label{eq:struct_rate0}
    \mathbb{E}\left[
        dN_{(\oA_{k-1},0)}(t)| \oA_{k}, \oL_{k}, D(\oA_{k-1},0)\geq t
    \right]=
    \mathbb{E}\left[
        dN_{(\oA_{k},0)}(t)|\oA_{k}, \oL_{k}, D(\oA_{k},0)\geq t
    \right]-
    A_{k} \tilde \beta_0 dt.
\end{equation}
We define \eqref{eq:struct_rate0} as the model $\mathcal{M}_0$. In other words, after lowering the most recent exposure to zero the counterfactual event rate would change by $A_k \tilde \beta_0$. In the context of our data analysis where $A_k\in\mathbb{R}^+$, positive $\tilde \beta_0$ implies that exposure is harmful as an increase in $A_k$ corresponds to an increase in the rate of events. It is also noted that the delayed effects of past exposure do not enter into \eqref{eq:struct_rate0} as we are only concerned with changing the most recent exposure.

The structural models then allow us to isolate each further delayed effect of exposure, $\tilde \beta_m$ ($m>0$). For the delayed effect $m$ periods prior, we consider an exposure history $\oA_{k-m}$ with more recent exposures set to zero, i.e., $A_{k-m+1}=\cdots=A_{k}=0$. In order to identify the effects, we nest the models by comparing two counterfactual rates that only differ by a single exposure value. For $t\in[k,k+1)$, we define model $\mathcal{M}_m$ by  
\begin{multline}\label{eq:struct_rate}
    \mathbb{E}\left[
        dN_{(\oA_{k-m-1},0)}(t)| \oA_{k-m}, \oL_{k-m}, D(\oA_{k-m-1},0)\geq t
    \right]\\
    =\mathbb{E}\left[
        dN_{(\oA_{k-m},0)}(t)|\oA_{k-m}, \oL_{k-m}, D(\oA_{k-m},0)\geq t
    \right]-
    A_{k-m} \tilde \beta_m dt.    
\end{multline}
In other words, $\mathcal{M}_m$ considers that in a stratum defined by $\{\oA_{k-m},\oL_{k-m}\}$ and $\{D(\oA_{k-m},0)\geq t\}$ the counterfactual event rate would be reduced by $A_{k-m} \tilde \beta_m$ if $A_{k-m}$ were set to zero.

\section{Estimation}
\subsection{Parametric exposure model estimators} \label{sec:parm}
The structural model under $\mathcal{M}_0$ provides a mean zero equation that can be used to form an unbiased estimation equation for the desired target parameter, i.e., $\tilde \beta_0$. 
Recalling the definition of $\mu_{km}(t),$ consider
$\mu_{k0}(t)=\mathbb{E}[A_{k}|\oA_{k-1},\oL_{k},D(\oA_{k},0) \geq t],$
$k \geq 0$ and $t \in [k,k+1),$
as a sequence of (short-term) time-varying exposure models. Importantly, when $t \in [k,k+1)$, the event $D(\oA_{k},0) \geq t$ is equivalent to $D \geq t;$ consequently, $\mu_{k0}(t)$ can be estimated using the observed data for each $k$. 
Moreover, following Assumption \ref{assump:causal}b, the counterfactual terminal event time $D(\oA_{k-1},0)$ is independent of $A_k$ when conditioning on a stratum defined by $\{\oA_{k-1},\oL_k\}$ and $D \geq t$, hence the proportion of individuals alive during $t\in[k,k+1)$ is equal to the proportion of individuals with $D \geq k;$ see Remark \ref{rem:indep}
for a detailed argument. 

Given a sequence of estimates for $\mu_{k0}(t)$, an estimator $\hat \beta_0$ 
can be defined as solving
\begin{equation}\label{eq:B0_est_eq}
    \sum_{i=1}^n\sum_{k=0}^K \int_{k}^{k+1}
    Y_i(t)
    \left\{A_{ik}- \hat \mu_{ik0}(t)\right\}
    \left\{dN_i(t)-  A_{ik}\beta_0dt\right\}
    =0
\end{equation}
for $\beta_0,$ where 
$\hat \mu_{ik0}(t)$ is the estimate of 
$\mu_{k0}(t)$ evaluated using the data
for subject $i.$ Under the proposed
structural model, the estimating equation
\eqref{eq:B0_est_eq} can be shown to be
unbiased for estimating $\tilde \beta_0$ 
if we replace $\hat \mu_{ik0}(t)$
with $\mu_{ik0}(t)$ for each $k$.
%
%
Similarly, the structural equation (\ref{eq:struct_rate}) can be used to construct an estimating equation for $\tilde \beta_m$. Let
$\Delta_{km}(t)=A_{k-m}-\mu_{km}(t),$
where $\mu_{km}(t)=\mathbb{E}[A_{k-m}|\oA_{k-m-1},\oL_{k-m},D(\oA_{k-m},0) \geq t].$
Arguing similarly to the case of $\hat \beta_0,$ we can first define
$\hat \beta_1$ as the solution to 
\begin{equation}\label{eq:Bm_est_eq_1}
    \sum_{i=1}^n\sum_{k=0}^K\int_k^{k+1} 
    Y_i(t) \hat w_{ik1}(t)
    \hat \Delta_{ik1}(t)
    \left\{dN_i(t)- 
        A_{ik} \hat \beta_0 dt - 
        A_{i(k-1)}\beta_1 dt
    \right\}
    =0,
\end{equation}
where, as discussed below, $\hat w_{ik1}(t)$ is an estimated weight to adjust for the proportion of individuals alive at time $t$ under the counterfactual exposure history ($m=1$) and
$\hat \Delta_{ik1}(t) = 
A_{i(k-1)}-\hat \mu_{ik1}(t).$
We then sequentially define $\hat \beta_m$ for any $m \geq 2$ as the solution to the estimating equation 
\begin{equation}\label{eq:Bm_est_eq}
    \sum_{i=1}^n\sum_{k=0}^K\int_k^{k+1} 
    Y_i(t) \hat w_{ikm}(t)
    \hat \Delta_{ikm}(t)
    \left\{dN_i(t)- 
        \sum_{j=0}^{m-1}A_{i(k-j)} \hat \beta_j dt - 
        A_{i(k-m)}\beta_m dt
    \right\}
    =0,
\end{equation}
where $\hat \beta_j, j=0,\ldots,m-1$ is the sequence of 
prior solutions. Similarly to \eqref{eq:B0_est_eq},
replacing the indicated estimates with the true
unknown functions $w_{ikm}(t)$ and $\mu_{ikm0}(t)$
leads to an unbiased sequence of estimating equations for the target parameters
$\tilde \beta_m.$

To understand the difference between \eqref{eq:B0_est_eq} and \eqref{eq:Bm_est_eq}, consider going from $dN_i(t)$ to \linebreak $dN_{i(\oA_{k-2},0)}(t)$ for $t\in[k,k+1)$: first, to go from $dN_i(t)$ to $dN_{i(\oA_{k-1},0)}$, we `blip' down by $A_k \tilde \beta_0$; next, the SNCURE defines the difference between $dN_{(\oA_{k-1},0)}$ and $dN_{(\oA_{k-2},0)}$ as $A_{k-1} \tilde \beta_1$. 

The weights $w_{ikm}(t)$ used to define
\eqref{eq:Bm_est_eq} ensure the weighted proportion of individuals with $D\geq t$, $t\in[k,k+1)$, is equal to the counterfactual proportion of individuals with $D(\oA_{k-m},0) \geq t$. In particular, for $m\geq 1,$ 
$w_{ikm}(t)$ is calculated by evaluating
\begin{equation}\label{eq:weight_def}
	w_{km}(t)= 
    \prod_{j=0}^{m-1}
    \exp\left\{
        A_{k-j}\boldsymbol\nu_{jk}(t)'\boldsymbol\alpha_m \right\}
\end{equation}
at $A_{k-j} = A_{i(k-j)},$ where $\boldsymbol\alpha_m=[\alpha_0,\ldots,\alpha_{m-1}]'$ is an $m$-length vector of parameters for the effect of exposure on the terminal event rate; and $\boldsymbol\nu_{jk}(t)$ is a vector of $j$ ones followed by $t-k$. For example, when $m=1$ and $j=0$, $\boldsymbol\alpha_m=\alpha_0$ and $\nu_{jk}(t)=t-k$; when $m=3$ and $j=2$, $\boldsymbol\alpha_m=[\alpha_0,\alpha_1,\alpha_2]'$ and $\nu_{jk}(t)=[1,1,t-k]'$. To get a better intuition for the weights, we consider the counterfactual scenarios for $t\in[k,k+1)$: if $A_k$ were instead set to zero, the number of individuals alive in $t$ is proportional to $\exp\{A_k\alpha_0(t-k)\}$; if both $A_k$ and $A_{k-1}$ were set to zero the number of individuals alive at time $t$ is proportional to $\exp\{A_{k-1}\alpha_0+A_{k-1}\alpha_1(t-k)+A_k\alpha_0(t-k)\}$ as the effect of $A_{k-1}\alpha_0$ has already occurred, while the effects $A_{k-1}\alpha_1(t-k)+A_k\alpha_0(t-k)$ are currently underway.

\begin{remark}
    The death rate parameters $\boldsymbol\alpha$ and the exposure models $\mu_{km}(t)$ can be estimated using the methods proposed by \citet{seaman2020adjusting}. For completeness, we have presented these methods in Supplementary Materials Section \ref{app:nuisance}.
\end{remark}

\begin{remark} \label{rem:indep}
    The independence condition  in Assumption \ref{assump:causal}b is based on conditioning on $D \geq k$, while the unbiasedness of the estimating equation requires conditioning on $D(\bar A_{k-1},0) \geq t$. We show that Assumption \ref{assump:causal}b is sufficient to build an unbiased estimating equation. We have
    \begin{align*}
      \mathbb{E}\left[
        dN_{(\oA_{k-1},0)}(t)\mid \oA_{k}, \oL_{k}, D(\oA_{k-1},0)\geq t
    \right] &=   \frac{\mathbb{E}\left[dN_{(\oA_{k-1},0)}(t)\mid \oA_{k}, \oL_{k}, D(\oA_{k-1},0)\geq k\right]}{p\{ D(\oA_{k-1},0)\geq t \mid \oA_{k}, \oL_{k}, D(\oA_{k-1},0)\geq k\}}\\
    &=\frac{\mathbb{E}\left[dN_{(\oA_{k-1},0)}(t)\mid \oA_{k-1}, \oL_{k}, D\geq k\right]}{p\{ D(\oA_{k-1},0)\geq t \mid \oA_{k-1}, \oL_{k}, D\geq k\}}.
    \end{align*}
    The first equality follows from Bayes' rule, $p\{D(\oA_{k-1},0)\geq t \mid \oA_{k}, \oL_{k}, dN_{(\oA_{k-1},0)}(t)=1 \}=1$ and $\mathbb{E}\left[dN_{(\oA_{k-1},0)}(t)\mid \oA_{k}, \oL_{k}\right]=\mathbb{E}\left[dN_{(\oA_{k-1},0)}(t)\mid \oA_{k}, \oL_{k}, D(\oA_{k-1},0)\geq k\right]$. The second equality follows from Assumption \ref{assump:causal}b and $D(\oA_{k-1},0)\geq k$ if and only if  $D\geq k$. The right-hand side of the equality is independent of $A_k$ which implies that the expected value of the proposed estimating equation will be zero.
\end{remark}

\subsection{Robust exposure model estimators}

The asymptotic linearity of the estimators proposed in Section \ref{sec:parm} relies on correctly specified parametric exposure models.  To allow for a nonparametric exposure model, we propose to also model the baseline rate function. For $t \in [k,k+1)$, let 
\[
\mathbb{E}\left[
        dN_{(\oA_{k-m-1},0)}(t)| \oA_{k-m}, \oL_{k-m}, D(\oA_{k-m-1},0)\geq t
    \right] \equiv d\eta_{km}(t,\oA_{k-m-1}, \oL_{k-m}).
\]
Then, following (\ref{eq:struct_rate1}), 
\begin{align}\label{eq:struct_rater}
    \mathbb{E}\left[
        dN_{(\oA_{k-m},0)}(t)|\oA_{k-m}, \oL_{k-m}, D(\oA_{k-m},0)\geq t
    \right] = d\eta_{km}(t,\oA_{k-m-1}, \oL_{k-m})+
    A_{k-m} \tilde \beta_m dt. 
\end{align}
For $m=0$, this additive rate model can be written in terms of the observed counting process as $\mathbb{E}\left[
        dN(t)|\oA_{k}, \oL_{k}, D\geq t
    \right] = d\eta_{k0}(t,\oA_{k-1}, \oL_{k})+ A_{k} \tilde \beta_0 dt$. For $m>0$, and similarly to the previous section,
    we first remove the exposure effects $\tilde\beta_0,\cdots, \tilde\beta_{m-1},$ allowing us to express this equation 
    in terms of the observed counting process as
\begin{multline}\label{eq:blippedrate1}
    \mathbb{E}\left[
        dN(t)- \sum_{j=0}^{m-1} A_{i(k-j)}\tilde\beta_j dt ~ \bigg| \oA_{k-m}, \oL_{k-m}, D(\oA_{k-m},0)\geq t
    \right]\\ = 
    d\eta_{km}(t,\oA_{k-m-1}, \oL_{k-m})+ A_{k-m}\tilde\beta_m dt.
\end{multline}    
In general, the potential counting process $dN_{(\oA_{k-m-1},0)}(t)$ is not observable; therefore, when $m > 0,$ one cannot model $\eta_{km} $ using the observed data. We overcome this issue by first modeling the blipped-down counting process. 
    
Specifically, following the same notational convention
that was used to define $\mu_{km}(t)=\mathbb{E}[A_{k-m}|\oA_{k-m-1},\oL_{k-m},D(\oA_{k-m},0) \geq t],$
let $d \rho_{km}(t)$
denote the conditional 
expectation of the left-hand side of \eqref{eq:blippedrate1} given $\{\oA_{k-m-1},\oL_{k-m},D(\oA_{k-m},0)\geq t\}.$ Then,
using \eqref{eq:blippedrate1}, we have
\begin{align} 
    \label{eq:rhokm}
    d\rho_{km}(t)
    = d\eta_{km}(t,\oA_{k-m-1}, \oL_{k-m})+ \mu_{km}(t) \tilde\beta_m dt.
\end{align}
Subtracting \eqref{eq:rhokm} from
both sides of \eqref{eq:blippedrate1} then leads
to 
\begin{align}
\label{eq:blippedrate2}
 \mathbb{E}\Bigl[
    dN(t)- \sum_{j=0}^{m-1} & 
    A_{i(k-j)}\tilde\beta_jdt~ 
    \bigg| \oA_{k-m}, \oL_{k-m}, D(\oA_{k-m},0)\geq t
    \Bigr] - d\rho_{km}(t) = \tilde\beta_m \Delta_{km}(t) dt,
    \end{align}
where we recall $\Delta_{km}(t) = 
A_{k-m}-\mu_{km}(t).$
The main advantage of this transformation is that it allows us to model both exposure model and $d\rho_{km}$ using appropriate nonparametric statistical learning methods. In Supplementary Materials Section \ref{app:sim}, we discuss one possible approach to estimate both
$d\rho_{km}(t)$ and $\mu_{km}(t)$. 

To avoid overfitting, and also to relax certain Donsker conditions that are often imposed in deriving asymptotic results for nonparametric models (see \S \ref{sec:asy}), we will leverage $V$-fold crossfitting to estimate
the $\tilde \beta_j$s.  
Specifically, we first split the data, uniformly at random, into $V$ mutually exclusive and exhaustive sets of size approximately $n V^{-1}$. For the $v^{th}$ fold, denote by $P_{n,v}^0$ the corresponding empirical distribution of the training sample, assumed to be of approximate size $n (1-V^{-1});$ and denote by $P_{n,v}^1$ the empirical distribution of the remaining validation sample, 
which is of approximate size $nV^{-1}.$ Finally, let $\hat \mu_{kmv}(t)$ and $d\hat \rho_{kmv}(t)$ be estimates of $\mu_{km}(t)$ and $d\rho_{km}(t)$ derived from the training sample
corresponding to the $v$th sample split. Then
we propose to use (\ref{eq:blippedrate2}) to 
estimate $\tilde \beta_m$ by solving
\begin{multline}\label{eq:robust_Bm_est_eq}
    P_{n,v}^1 \sum_{k=0}^K\int_k^{k+1} 
    Y(t)\hat w_{kmv}(t)
    \hat\Delta_{kmv}(t)\\
    \times\left\{dN(t)- 
        \sum_{j=0}^{m-1}A_{(k-j)}\hat \beta_j dt - d\hat \rho_{kmv}(t)-
        \hat\Delta_{kmv}(t)\beta_m dt
    \right\}=0,
\end{multline}
where $m=0,\cdots,k-1$, $\hat w_{k0v}(t) = 1,$ $\hat\Delta_{kmv}(t)=A_{k-m}-\hat\mu_{kmv}(t)$ and $\sum_{j=0}^{0-1}A_{i(k-j)}\hat \beta_j dt=0$. The nuisance parameter estimate $\hat w_{kmv}(t)$ can be obtained by applying the  method described in Section \ref{sec:parm} to the training sample for the $v$th sample split; however, since the model used
is parametric, one can also use $\hat w_{km}(t)$ as
previously derived from
the full sample.

\subsection{Asymptotic results}
\label{sec:asy}
We consider both parametric and nonparametric estimators for $\mu_{km}(t)$ for $m \geq 0.$ Below and in the proofs, we use empirical process notation $P_n f = n^{-1} \sum_{i=1}^n f(x_i)$ and $P_0 f = E_{P_0} f$. 

\begin{assumption}[Parametric models for $\mu_{km}$] 
\label{ass:parm}
    We impose the following assumptions:
    \begin{itemize}
        \item[(a)]  For each $k,$ suppose  that $\mu_{km}(t)$ is  known to lie in a parametric model $\{\mu_{km\theta_k}:\theta_k \in \Theta_k\}$ (i.e., $\mu_{km}=\mu_{km \tilde\theta_{k}}$ where $\tilde\theta_{k}$ denote the true parameter value). Let $\hat\mu_{km}(t)=\mu_{km \hat\theta_{k}}(t)$ be a corresponding consistent estimator. Also, define \linebreak $\dot\mu_{km \tilde\theta_{k}}(t)= \left.\frac{\partial}{\partial \theta_k}\mu_{km \theta_k}(t)\right|_{\theta_k=\tilde\theta_{k}}$ and assume $\dot\mu_{km \theta_k}(t)$ is continuous and bounded for $\theta_k$ in an open neighborhood of $\tilde \theta_k.$ 
        Finally, for each $k$, assume $\hat\theta_{k}$ is  an asymptotically linear estimator of $\tilde\theta_{k}$ with influence function $\phi_{km 
        \tilde\theta};$ that is, $\hat\theta_{k}-\tilde\theta_{k} = (P_n-P_0)\phi_{km \tilde\theta}+o_p(n^{-1/2})$.
        
        \item[(b)] For each $k$, $w_{km}(t)$ is  known to lie in a parametric model $\{w_{km\alpha_k}:\alpha_k \in \mathcal{A}_k\}$ (i.e., $w_{km}=w_{km \tilde\alpha_k}$ where $\tilde\alpha_{k}$ denote the true parameter value)). Let $\hat w_{km}(t)=w_{km\hat\alpha_{k}}(t)$ be a corresponding consistent estimator. Also, define $\dot w_{km \tilde\alpha_{k}}(t)=\left.\frac{\partial}{\partial \alpha_k}w_{km \alpha_k}(t)\right|_{\alpha_k=\tilde\alpha_{k}}$ and assume $\dot w_{km \alpha_k}(t)$ is continuous and bounded for $\alpha_k$ in an open neighborhood of $\alpha_{k0}$. Finally, for each $k$, assume $\hat\alpha_{k}$ is  an asymptotically linear estimator of $\tilde\alpha_{k}$ with influence function $\phi_{km \tilde\alpha};$ that is, $\hat\alpha_{k}-\tilde\alpha_{k} = (P_n-P_0)\phi_{km\tilde\alpha}+o_p(n^{-1/2})$. 
    \end{itemize}
\end{assumption}

\begin{theorem}[Parametric models for $\mu_{km}$]
\label{thm:parmodels}
  Let  $\mu_{km}(t)$ be a parametric time-varying exposure model for $m \geq 0.$ Under Assumption \ref{ass:parm}.a, the estimator $\hat \beta_0$ defined as a solution to (\ref{eq:B0_est_eq}) is asymptotically linear,
  and 
  \[
  \hat \beta_0 -  \tilde\beta_0 = 
  (P_n-P_0) \left\{ \frac{ \sum_{k=0}^K  \int_{k}^{k+1} Y(t)\left\{A_{k}-\mu_{k0}(t)\right\}dN_{0}(t) - \sum_{k=0}^K \phi_{k0 \tilde\theta} \gamma_{k0} }{P_0 g( \mu)} 
  \right\} + o_p(n^{-1/2}),
  \]
where $\gamma_{k0} = 
P_0 \left[ 
\int_{k}^{k+1}
\dot\mu_{k0 \tilde\theta_{k}}(t)
Y(t) d N_0(t) \right]$, $g(\mu) = \sum_{k=0}^K \int_{k}^{k+1} Y(t)\left\{A_{k}-\mu_{k0}(t)\right\}A_{k}dt$ and $dN_{0}(t) = dN(t)-\tilde\beta_0A_{k}dt$.

Let $w_{km}(t)$ be the weight functions as defined in (\ref{eq:weight_def}) and
let
$\Delta_{km}(t) = 
A_{k-m}-{\mu}_{km}(t)$
for
some $m \geq 1.$ Let $\phi_{\tilde\beta_j}$ be the influence function for $\hat{\beta}_{j}$
for $j \leq m-1;$ i.e., $\hat{\beta}_{j}-\tilde\beta_{j}=(P_n-P_0)\phi_{\tilde\beta_j} + o_p(n^{-1/2})$. Then, under Assumption \ref{ass:parm}, and defining
$dN^{(m)}_0(t)=dN(t)-(\sum_{j=0}^{m-1} A_{k-j}\tilde\beta_{j}+A_{k-m}\tilde\beta_m)dt$ for $m \geq 1,$ we have 
\begin{align*}
\hat{\beta}_m  - \tilde\beta_m
= \frac{1}{P_0 g(\mu,w)} & \times
(P_n-P_0) \sum_{k=0}^{K} \biggl[
 \int_{k}^{k+1}Y(t){w}_{km}(t)
 \Delta_{km}(t)
 dN^{(m)}_0(t) \\
& + 
\{ \phi_{km \tilde \alpha} 
 \, \xi^{(2)}_{km \tilde \alpha}
- \phi_{km \tilde \theta}  \,
\xi^{(1)}_{km \tilde \theta_k} \}
-
 \sum_{j=0}^{m-1} \phi_{\tilde\beta_j} 
    \int_{k}^{k+1}  \gamma_{kmj}(t) dt 
\biggr]
+  o_P(n^{-1/2}),
\end{align*}
where
$\xi^{(1)}_{km \tilde \theta_k} = P_0 \int_{k}^{k+1}Y(t) {w}_{km}(t) \dot\mu_{km\tilde \theta_{}}(t) dN^{(m)}_0(t)$, 
$\xi^{(2)}_{km \tilde \alpha} = P_0 \int_{k}^{k+1}Y(t) \dot w_{km \tilde \alpha}(t)  \Delta_{km}(t)
 dN^{(m)}_0(t)$, 
 $\gamma_{kjm}(t) =
 P_0  Y(t) w_{km}(t)
 \Delta_{km}(t)
A_{k-j},$
 and
 $g(\mu,w)=\sum_{k=0}^{K}\int_{k}^{k+1} Y(t){w}_{km}(t) \Delta_{km}(t)
 A_{k-m}dt$.

\end{theorem}
We note that the first part of 
Theorem \ref{thm:parmodels}
provides the stated influence
function for $\hat \beta_0,$
namely $\phi_{\tilde\beta_0}.$
This influence function can be
substituted into the influence
function for
$\hat{\beta}_m  - \tilde\beta_m$
above when $m =1;$ along with 
$\phi_{\tilde\beta_0},$
the resulting influence
function, or $\phi_{\tilde\beta_1},$
can then be
substituted into the influence
function given 
for $\hat{\beta}_m  - \tilde\beta_m$
when $m =2;$ and, so on.

\begin{theorem}[Nonparametric models for $\mu_{km}$ and $d\rho_{km}$] 
\label{thm:nonparmodels}
  Let  $\mu_{k0}(t)$ be a time-varying exposure model,  $w_{km}(t)$ be the weight functions as defined in (\ref{eq:weight_def}) and $d\rho_{km}$ be the conditional mean function defined in (\ref{eq:rhokm}). Under Assumptions \ref{ass:parm}.b and \ref{ass:nonparm-new} (see the appendix), the estimator $\hat \beta_0$ defined as a solution to (\ref{eq:robust_Bm_est_eq}) is asymptotically linear:
\begin{align*}
    \hat \beta_0-\tilde\beta_0
    =
(P_{n} - P_0)
\left[
\frac{\sum_{k=0}^K \int_{k}^{k+1} Y(t) \Delta_{k0}(t) 
\{dN(t) - d\rho_{k0}(t)-\Delta_{k0}(t)\tilde\beta_0 dt\}}{P_0 g(\mu)}
\right]
+ o_P(n^{-1/2}).
\end{align*}
Let $w_{km}(t)$ be the weight functions as defined in (\ref{eq:weight_def}) and
let
$\Delta_{km}(t) = 
A_{k-m}-{\mu}_{km}(t)$
for
some $m \geq 1.$ Let $\phi^{\star}_{\tilde\beta_j}$ be the influence function for $\hat{\beta}_{j}$
for $j \leq m-1;$ i.e., $\hat{\beta}_{j}-\tilde\beta_{j}=(P_n-P_0)\phi^{\star}_{\tilde\beta_j} + o_p(n^{-1/2})$. Then, 
defining
$dM^{(m)}_0(t)=dN(t)- d \rho_{km}(t) -(\sum_{j=0}^{m-1} A_{k-j}\tilde\beta_{j}+\Delta_{km}(t)\tilde\beta_m)dt$ for $m \geq 1,$ we have 
\begin{align*}
\hat{\beta}_m  - \tilde\beta_m
= \frac{1}{P_0 g(\mu,w)} & \times
(P_n-P_0) \sum_{k=0}^{K} \biggl[
 \int_{k}^{k+1}Y(t){w}_{km}(t)
 \Delta_{km}(t)
 dM^{(m)}_0(t) \\
& + 
\phi_{km \tilde \alpha} 
 \, \xi^{(3)}_{km \tilde \alpha}
-
 \sum_{j=0}^{m-1} \phi^{\star}_{\tilde\beta_j} 
    \int_{k}^{k+1}  \gamma_{kmj}(t) dt 
\biggr]
+  o_P(n^{-1/2}),
\end{align*}
where
$\xi^{(3)}_{km\tilde\alpha}  =
P_0
\int_{k}^{k+1} 
Y(t) \dot w_{km \tilde\alpha}(t)  \Delta_{km}(t) dM^{(m)}_0(t),$
 $\gamma_{kjm}(t) =
 P_0  Y(t) w_{km}(t)
 \Delta_{km}(t)
A_{k-j},$

 and
 \[
 g(\mu,w)=\sum_{k=0}^{K}\int_{k}^{k+1} Y(t){w}_{km}(t) \Delta_{km}(t)
 A_{k-m}dt.
 \]
\end{theorem}

\subsection{Variance of $\beta_m$}
We use a bootstrap approach to estimate the variance of the short-term and delayed effects, $\beta_m$. In particular, we re-sample individuals with replacement, retaining their entire observational timeline. Next, for each bootstrap sample, $r=1,\ldots,R$ we estimate all components of the model \{$\hat \mu_{km}^{(r)}, \hat \alpha_m^{(r)}, d\hat \rho_{km}^{(r)}, \hat \beta_m^{(r)}$\} and calculate $\text{Var}\{\hat \beta_m\} =R^{-1} \sum_{r=1}^R (\hat \beta_m^{(r)}-\overline{\beta}_m)^2$, where $\overline{\beta}_m=R^{-1}\sum_{r=1}^R\hat \beta_m^{(r)}$ is the bootstrap mean.

\subsection{Non-administrative censoring}\label{sec:censor}
We consider non-administrative censoring at some time $C<\tau$. 
Let $N_C(t) = I(C \leq t)$ denote the censoring counting process with $\lambda_c(t,\bar A_{\floor{X}},\bar L_{\floor{X}},X )=E\{dN_C(t) \mid C \geq t, \bar A_{\floor{X}},\bar L_{\floor{X}},X >t,X \}$ where $X = D \wedge C$.   Under two conditions our original model holds when non-administrative censoring occurs \citep{Vansteelandt2016RevisitingConfounding}. First, it is assumed that future exposures do not impact the censoring process (i.e., $\lambda_c(t,\bar A_{\floor{X}},\bar L_{\floor{X}},X )= \lambda_c(t,\bar A_{\floor{t}},\bar L_{\floor{t}} )$). Second, current and past exposures, $\oA_k$, are assumed to be uncorrelated with the censoring process (i.e., $\lambda_c(t,\bar A_{\floor{t}},\bar L_{\floor{t}} )= \lambda_c(t, L_0 )$). The second assumption can be relaxed by replacing $w_{km}(t)$ with $w_{km}(t)\times w_{km}^{C}(t)$, where \linebreak $w_{km}^{C}(t) = \exp\left\{\int_k^t \lambda_c(t,\bar A_{\floor{t}-m},\bar L_{\floor{t}-m} ) dt\right\}$ is the inverse probability of remaining uncensored at time $t$.

\section{Simulation}
To understand the finite sample performance and limitations of the proposed methods, we developed a simulation study to estimate the short-term and delayed effects of time-varying exposures on a recurrent event process. We demonstrate empirical evidence that our estimation procedures are root-$n$ consistent and that bootstrap standard errors produce valid confidence intervals with nominal coverage of the true exposure effects. Finally, we explore the extent to which misspecification of the parametric model introduces bias into our estimates.

Our simulation study considers the effect of continuous time-varying exposures on recurrent events with a correlated death process. For each observation, $i$, we sample a time-series of two covariates, $\{L_{i1k}, L_{i2k}\}$, $k=-4,\ldots,30$ where $k<0$ provides a measure of baseline factors potentially related to future events. The covariates $\oL_{iK}$ are generated as multivariate normal with $\mathbb{E}(L_{ijk})=0$, Cov$(L_{ijk},L_{ij'k'})=\mathbb{I}_{(j=j')}\sigma_j^2\Sigma(k,k')$, with $\sigma_1=0.2$, $\sigma_2=1$, and $\Sigma(k,k')=0.95^{|k-k'|}$. We also sample a fixed frailty variable for each observation, $Q$, from an exponential distribution with $\mathbb{E}(Q)=0.2$. Continuous exposures are generated for two distinct simulation scenarios: first by a linear function of covariates with $A_{ik}\sim\mathcal{N}(L_{i1k}+0.5 L_{i2k},1)$, and second by a complex function of covariates $A_{ik}\sim\mathcal{N}(2L_{i1k}^2+2L_{i1k}\cdot|L_{i2k}-1|,1)$; exposures are normalized such that $A_{k}\in[0,1]$.

We generate recurrent event $j=1,2,\ldots$ at time $T_j$ by drawing $U\sim\text{Uniform}(0,1)$ and solving for $t$ the conditional survivor function, $U=\mathbb{P}(T_j>t|T_{j-1},\ldots,T_0=0)=\exp\{-\int_{T_{j-1}}^t \lambda(v) dv\}$,
where $\lambda(v)$ is the event rate function and $T_0\equiv 0$. For $t\in[k,k+1)$ we define the event rate as an additive function of exposures plus a baseline intensity of the individual-level covariates, $\lambda(t) = Y(t)\left\{\sum_{m=0}^4 A_{k-m} \tilde \beta_m + c\cdot \eta(\oL_k)dt \right\}$,
where $\tilde\beta_0=0.1, \tilde\beta_1=0.05, \tilde\beta_2=0.025$, and $\tilde\beta_3=\tilde\beta_4=0$ define the exposure effects, $Y(t)$ is a survival indicator, $\eta(\oL_k)dt= Q\cdot \exp\{L_{1k}+L_{1k}^2+L_{2k}-1\}$ is the baseline intensity, and $c$ is set such that $\text{Var}\{c\cdot\eta(\oL_k)dt\}/ \text{Var}(\sum_{m=0}^4 A_{k-m}\tilde\beta_m)=100$. By scaling the baseline intensity by $c$ we ensure that the confounding bias is large relative to the exposure effect, which represents the more realistic setting found in our data analysis. We generate a terminal event $D$ from the rate function $\tilde{\lambda}(t)=0.02A_{k}+0.01A_{k-1}+Q\exp\{L_{1k}+L_{2k}-1\}$ and non-administrative censoring $C$ at rate $\tilde{\lambda}^C(t)=0.2\exp\{L_{1k}+L_{2k}-1\}$.

We simulate with $n=\{2000,5000\}$ observations and 100 replicates for both linear and complex exposure scenarios. To estimate the parameters $\tilde {\boldsymbol\beta}$ we compare three approaches, two only estimating the exposure model and one with the proposed robust estimator. The exposure model, $\mathbb{E}[A_{k-m}|\oA_{k-m-1},\oL_{k-m}, D(\oA_{k-m},0)\geq t]$, is estimated by two different approaches: 1) a parametric model or 2) a nonparametric approach. For the robust estimator, we fit the exposure model and nuisance parameter $d\rho_{km}(t)$ using the SuperLearner ensemble modeling approach, combining results from both a parametric model and a machine learning gradient boosting algorithm \citep{Ke2017LightGBM:Tree} fit with several combinations of hyperparameters. Specifically, we fit a separate model for each combination of $k$ and $m$ and allowed for time-varying  exposure using the pseudo-data approach with 5 equally spaced time bins for each period $[l,l+1)$, $l=0,\cdots,K$ (see Supplementary Materials Section \ref{app:sim}). All parameters ($\mu_{km}(t)$ and $d\rho_{km}(t)$) are estimated using V-fold crossfitting where the sample is split into five distinct subsets and the estimates of each subset for $\mu_{km}(t)$ or $d\rho_{km}(t)$ are based on a fit of the model to the remaining subsets. The competing risk parameters, $\alpha_m$, are always fit using
the previously described parametric approach. To our knowledge, no other methods for recurrent events exist that are comparable to our proposed approach for the described data generating mechanisms. Supplementary Materials Section \ref{app:sim} provides additional details on the estimation of $\mu_{km}(t)$ and $d\rho_{km}(t)$. 

We compare the models by $\sqrt{n}$-bias, standard error based on 200 bootstrap estimates, and coverage of the true parameters by 95\% confidence intervals. Results for the simple exposure model are given in Table \ref{tab:sim_scen1}. Under the simple exposure model scenario, all methods have low $\sqrt{n}$-bias, which is similar or smaller with increased sample size validating the asymptotic theory for consistency of our estimators. Compared to the parametric exposure model, estimator standard errors are slightly larger using a nonparametric exposure model approach; the standard errors are generally smallest for the robust estimator with nonparametric exposure model and nuisance parameter estimators. All methods cover the true parameter values with rates reasonably close to the nominal 95\% coverage rate although the robust estimator coverage appears slightly lower at later lags.

\begin{table}[!ht]
\centering\scriptsize
\setstretch{1.2}
\caption{Simulation results.}
\label{tab:sim_scen1}
\begin{tabular}[t]{rrrrrrcrrrrr}
\toprule[2pt]
&\multicolumn{5}{c}{$n=2000$}&&\multicolumn{5}{c}{$n=5000$}\\
\cmidrule{2-6}\cmidrule{8-12}
\multicolumn{1}{l}{\textbf{Simple exposure model}}&\multicolumn{1}{c}{$\beta_0$} &\multicolumn{1}{c}{$\beta_1$} &\multicolumn{1}{c}{$\beta_2$}  
&\multicolumn{1}{c}{$\beta_3$} &\multicolumn{1}{c}{$\beta_4$} & &\multicolumn{1}{c}{$\beta_0$} &\multicolumn{1}{c}{$\beta_1$} &\multicolumn{1}{c}{$\beta_2$} &\multicolumn{1}{c}{$\beta_3$} &\multicolumn{1}{c}{$\beta_4$} \\
\midrule
\multicolumn{5}{l}{\textit{$\sqrt{n}$-Bias}}\\
Parametric &  0.03 & -0.21 & -0.16 & -0.16 & -0.06 && -0.13 & -0.21 & 0.04 & -0.31 & -0.10\\
Nonparametric &  0.23 & -0.19 & -0.07 & -0.22 & -0.09 &&  0.00 & -0.14 & 0.02 & -0.34 & -0.10\\
Nonparametric Robust & -0.34 & -0.44 & -0.21 & -0.05 &  0.04 && -0.59 & -0.34 & -0.11 & -0.18 &  0.00\\

\addlinespace
\multicolumn{5}{l}{\textit{Standard Error} ($\times100$)}\\
Parametric & 2.49 & 2.52 & 2.57 & 2.57 & 2.58 && 1.60 & 1.64 & 1.65 & 1.65 & 1.66\\
Nonparametric & 2.67 & 2.68 & 2.68 & 2.70 & 2.66 && 1.62 & 1.66 & 1.66 & 1.68 & 1.68\\
Nonparametric Robust & 2.60 & 2.31 & 2.16 & 2.06 & 1.90 && 1.49 & 1.43 & 1.38 & 1.33 & 1.30\\

\addlinespace
\multicolumn{5}{l}{\textit{Coverage}}\\
Parametric & 0.97 & 0.93 & 0.95 & 0.95 & 0.95 && 0.94 & 0.98 & 0.91 & 0.93 & 0.96\\
Nonparametric & 0.92 & 0.91 & 0.90 & 0.96 & 0.96 && 0.93 & 0.98 & 0.92 & 0.94 & 0.93\\
Nonparametric Robust & 0.94 & 0.91 & 0.85 & 0.87 & 0.92 && 0.91 & 0.90 & 0.87 & 0.85 & 0.93\\

\addlinespace
\multicolumn{1}{l}{\textbf{Complex exposure model}}\\
\midrule
\multicolumn{5}{l}{\textit{$\sqrt{n}$-Bias}}\\
Parametric & -0.69 & -0.36 & -0.17 & -0.19 &  0.15 && -0.85 & -0.45 & -0.55 &  0.08 & 0.22\\
Nonparametric & -0.21 & -0.17 & -0.04 & -0.23 &  0.05 && -0.10 & -0.07 & -0.38 & -0.01 & 0.10\\
Nonparametric Robust & -0.63 & -0.41 & -0.20 & -0.14 & 0.06 && -0.49 & -0.28 & -0.37 &  0.08 & 0.14\\

\addlinespace
\multicolumn{5}{l}{\textit{Standard Error} ($\times100$)}\\
Parametric & 2.78 & 2.79 & 2.73 & 2.71 & 2.71 && 1.81 & 1.80 & 1.81 & 1.80 & 1.78\\
Nonparametric & 3.09 & 3.09 & 3.06 & 3.03 & 3.04 && 1.89 & 1.91 & 1.94 & 1.96 & 1.96\\
Nonparametric Robust & 2.98 & 2.72 & 2.50 & 2.31 & 2.20 && 1.75 & 1.67 & 1.63 & 1.57 & 1.52\\

\addlinespace
\multicolumn{5}{l}{\textit{Coverage}}\\
Parametric & 0.91 & 0.93 & 0.96 & 0.93 & 0.94 && 0.86 & 0.93 & 0.91 & 0.97 & 0.97\\
Nonparametric & 0.92 & 0.90 & 0.96 & 0.92 & 0.89 && 0.95 & 0.97 & 0.94 & 0.99 & 0.96\\
Nonparametric Robust & 0.91 & 0.93 & 0.92 & 0.91 & 0.87 && 0.92 & 0.92 & 0.92 & 0.96 & 0.94\\

\bottomrule[2pt]
\end{tabular}
\end{table}

Under the complex exposure model scenario (Table \ref{tab:sim_scen1}), the parametric approach tends to have larger bias across all parameters which appears to increase slightly with $n$. The other estimators show evidence of consistency with similar or smaller bias at larger $n$. We find that an exposure model only approach with nonparametric estimator has the largest standard errors while the robust approach generally has the smallest standard errors, similar to the first simulation scenario. Coverage of all models are near or slightly below the nominal level of $0.95$. Additional simulation results regarding power and type I error rate are given in Supplementary Materials Section \ref{app:sim}. Power for the robust model is similar to other approaches when $n=2000$ and uniformly highest when $n=5000$. Type I errors are near the expected rate of 0.05 in the $n=5000$ setting and occasionally higher for the smaller sample size setting.

\section{Effect of time-varying PM$_{2.5}$ exposure on recurrent CVD hospitalizations}

We apply our proposed method to estimate the effects of exposures to time-varying (monthly) PM$_{2.5}$ on recurrent CVD hospitalizations. Our longitudinal data set is based on 10,438,899 person-months of observation representing a 10\% sample ($n= 299,661$) of Medicare Fee-for-Service (FFS) beneficiaries living in the Northeast region of the United States (see Supplementary Materials Section \ref{app:data} for additional details). We considered a baseline period of $M=6$ months following each individual's entry into Medicare FFS (at age 65) and administrative censoring after $\tau=30$ months of follow-up. During follow-up, we recorded 7,814 primary-cause CVD hospitalizations, including diagnoses of myocardial infarction, atrial fibrillation, or cardiac arrest. In addition, our analysis considered competing risks for 11,565 individuals who died during follow-up. Although non administrative censoring was present (12.7\% of the sample at an average age of 66.7, IQR 66.1-67.3) due to leaving Medicare FFS or missing covariates, we do not believe that this was related to exposure to PM$_{2.5}$.

Following our simulation study, we used three estimators to estimate the short-term and delayed impact that PM$_{2.5}$ has on the rate of CVD hospitalization. Specifically, we estimated $\tilde\beta_m$ using two versions of the exposure model estimator (Parametric and Nonparametric) as well as the robust estimator (Nonparametric Robust). As covariates in the exposure model, we included ZIP code characteristics from the US census (see Supplementary Materials Section \ref{app:data}), centroid latitude and longitude, average monthly maximum temperature and minimum relative humidity, average monthly NO$_2$ and ozone concentrations, and indicators for state, month, and year. For the estimation of $d\rho_{km}(t)$, $t\in[k,k+1)$, we also included individual-level covariates: race, sex, Medicaid eligibility, and the number of CVD events during the period $[0,k)$. The confidence intervals are based on standard errors derived from 100 bootstrap samples. See Supplementary Materials Section \ref{app:data} for complete details on the fitting algorithms and the specification of the parametric model.

Estimates $\hat \beta_m$ for each of the three modeling approaches are given in Figure \ref{fig:alpha_beta_est}. Although the estimators are overall consistent in their results, we note several differences. First, the parametric estimator shows less smoothness across the lag periods defined by $\beta_m$. Second, the nonparametric exposure-model-only and nonparametric robust estimators have a lower variance compared to the parametric estimator. Third, when comparing the cumulative effects (i.e., $\sum_{m=0}^\ell\hat \beta_m$ for given lag $\ell$), we note an overall trend of increasing CVD hospitalization rates due to longer-term increases in PM$_{2.5}$ exposure. Finally, the effects of the nonparametric estimator are slightly larger compared to the other two estimators.

\begin{figure}[!ht]
    \centering
    \subcaptionbox{Effects of a 10$\mu g/m^3$ increase in monthly average PM$_{2.5}$ exposure on the change in rate of CVD hospitalizations (y-axis) per 100,000 Medicare beneficiaries. Columns describe the estimators used. The top row shows the individual $\hat \beta_m$ parameters for lag $m$ (x-axis). The bottom row describes the cumulative effect (i.e., $\sum_{m=0}^\ell \hat \beta_m$ where $\ell$ is the length of lagged exposure increases) for a consecutive increase in exposure beginning in the immediate month ($x=0$). Shaded area represents bootstrap 95\% confidence intervals.\label{fig:alpha_beta_est}}{\includegraphics[width=0.8\textwidth]{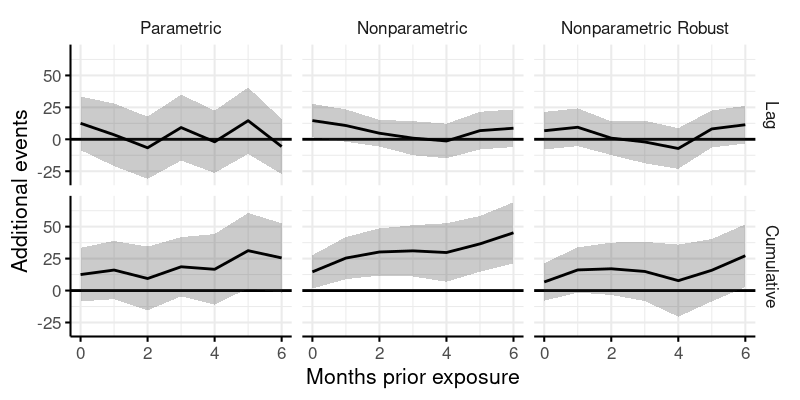}}
    \subcaptionbox{Estimated decrease in primary cause CVD hospitalizations (95\% confidence interval in shaded area) over 30 months if maximum PM$_{2.5}$ were limited to 12 or 9 $\mu g/m^3$.\label{fig:add_events}}{\includegraphics[width=0.8\textwidth]{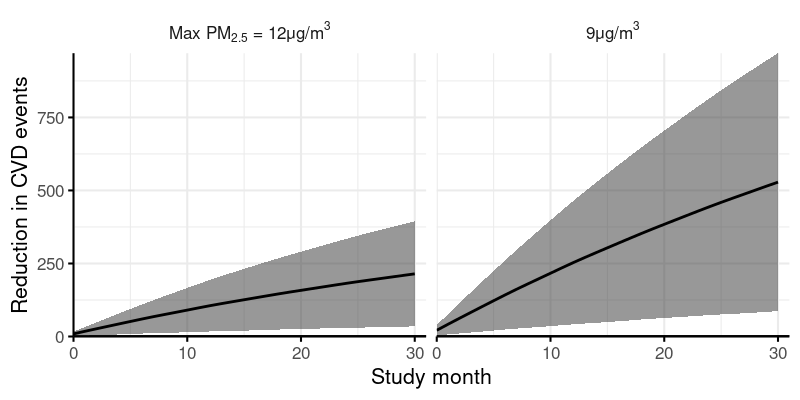}}
    \caption{Effects of PM$_{2.5}$ exposure on CVD hospitalizations.}
\end{figure}

Using the robust nonparametric estimator as the standard for inference, we estimate that a 10-$\mu g/m^3$ increase in monthly average PM$_{2.5}$ concentration over the current and six prior months would correspond to an additional 27.3 (95\% CI: 2.7-51.9) CVD-related hospitalizations per 100,000 Medicare beneficiaries. We note that the greatest effects of PM$_{2.5}$ on CVD hospitalization rates are due to same-month and 1 month prior exposure, as well as exposures from 5-6 months prior. Although individual lag effects, $\hat \beta_m$, are not significantly different from null, there is evidence that continued exposure over the course of several months (i.e., $\sum_{m=0}^6\hat \beta_m$) corresponds to a significant increase in the rate of CVD hospitalizations. Additional results of the analysis and data limitations are given in the Supplementary Materials \ref{app:data}.

\subsection{Additional events incurred}
Based on estimates $\hat{\beta}_m$, $m=0,\ldots,M$ we propose an estimate for the number of events incurred due to an intervention in the maximum allowed PM$_{2.5}$ exposure among our sample. Let $N_{i(a^*)}(t)$ indicate the counterfactual number of events for individual $i$ at time $t\in[k,k+1)$ under exposure history $a^*=\{A_{-M}^*,\ldots,A_{k}^*\}$ where $A_{k'}^*=A_{k'}\wedge a$ is the minimum of the observed exposure or a counterfactual exposure $a$. For $t\in[k,k+1)$ we calculate,
\begin{multline}\label{eq:add_events}
    \sum_{i=1}^n\hat{\mathbb{E}}\{N_i(t)-N_{i(a^*)}(t)|D_i\geq t\} =
    \sum_{i=1}^n\sum_{k=0}^{\floor{t}} \int_k^{(k+1)\wedge t} Y_i(t)\left\{\sum_{m=0}^M (A_{i(k-m)}- A_{i(k-m)}^*)\hat{\beta}_mdt\right\}.
\end{multline}

The confidence intervals in the estimates are obtained using the bootstrap estimates of $\hat \beta_m$. Note that \eqref{eq:add_events} only considers changes in the number of events prior to observed death or censorship of the individual $i$. It is reasonable to assume that when $A_k$ also has an effect on mortality, additional events may be observed given a longer lifetime; however, accomplishing this requires consistent estimation of the nuisance function $d\eta_{km}(t,\oA_{k-m-1},\oL_{k-m})$.

Figure \ref{fig:add_events} shows the estimated decrease in events based on the robust estimator under counterfactual exposure scenarios where the monthly average exposure is limited to 12 or 9 $\mu g/m^3$ PM$_{2.5}$ (the most recent as well as newly updated EPA National Ambient Air Quality Standards of annual average exposure for PM$_{2.5}$). Throughout the course of the 30 month study period and assuming that each individual survived and was uncensored for the same period of time regardless of exposure, we estimate that limiting exposure to 12 or 9 $\mu g/m^3$ PM$_{2.5}$ would have resulted in 214 (95\% CI 35-394) or 528 (95\% CI 87-970) fewer primary CVD hospitalizations, respectively. With respect to the total number of hospitalizations for CVD observed, we estimate that limiting monthly average exposure to $9\mu g/m^3$ PM$_{2.5}$ over the same time period would have resulted in a 6.8\% decrease in the total number of observed events during this time frame. Extrapolating to the entire Northeast Medicare FFS population would imply that limiting monthly average exposure to 9$\mu g/m^3$ would have eliminated 5,280 primary-cause CVD hospitalizations for individuals between ages 65.5 and 68. A recent study found that the average total hospitalization costs for CVD were \$13,100 for atrial fibrillation, \$29,500 for myocardial infarction, and \$24,400 for cardiac arrest \citep{tajeu2024cost}. Hence, a conservative estimate, using the average cost of hospitalization for atrial fibrillation, implies that limiting monthly average PM$_{2.5}$ to a maximum of 9$\mu g/m^3$ over a 10-year period would reduce associated healthcare costs for Medicare FFS beneficiaries in the Northeast US by \$277 million.

From a regulatory perspective, such as the PM$_{2.5}$ Regulatory Impact Analysis from the Environmental Protection Agency, the focus on a single health endpoint can lead to underestimates of the total health and welfare benefits resulting from lower air pollution standards \citep{U.S.EnvironmentalProtectionAgency2022RegulatoryMatter}. Some studies have estimated the association between long-term exposure to air pollution and CVD progression. These studies have considered single hospital readmission as the outcome and focused on exposure measured at a single time point \citep{Koton2013CumulativeOutcomes, Zhang2023AssociationsDiseases} or exposure that varied over time without considering delayed effects \citep{Zanobetti2007ParticulateInfarction, Liu2018FineCities, VonKlot2005AmbientCities}.  Our framework allows a more holistic understanding of the cumulative effects of time-varying exposure while also accounting for competing risks of disease-related death. Our findings underscore the significant impact of fine particulate matter exposure on recurrent cardiovascular hospitalizations, providing critical evidence for policymakers and public health officials aiming to mitigate air pollution’s adverse effects on vulnerable populations.

\section{Discussion}\label{sec:discussion}
Our SNCURE framework allows us to estimate the effects of time-varying exposures on recurrent event outcomes in the presence of a correlated terminal event. We develop two estimators: a parametric method that only requires estimation of an exposure model, and a robust method that relaxes parametric assumptions but also requires estimation of a nuisance function. Our robust estimator allows for the application of most nonparametric statistical learning methods to model both exposure and nuisance functions. This flexible approach alleviates the unrealistic requirement to correctly specify these parametric models. For both estimators, we demonstrate, through rigorous proof and empirical simulation, they consistent and asymptotically normal. Finally, we provide the R software package \textit{sncure} to apply our methods in a real data scenario.



Our data application to a longitudinal cohort of Medicare beneficiaries in the Northeast US is the first to investigate the short-term and delayed effects of PM$_{2.5}$ on CVD hospitalization rates. The analysis found that increased exposure to PM$_{2.5}$ during the current and prior six months is related to increased rates of CVD hospitalization. A limitation of our study is that we did not explicitly model spatial dependence. However, because we adjusted for numerous ZIP code–level predictors of exposure and mortality, including demographics, socioeconomic status, and meteorological variables, any remaining spatial dependence is likely minimal. Furthermore, assuming independent residuals across ZIP codes, remaining spatial dependence would not affect our point estimates but could impact the estimated influence function-based variances. It would be valuable to incorporate a specialized blocked bootstrap that captures spatial dependence \citep{wilks1997resampling,lahiri2013resampling}. Alternatively, an m-out-of-n bootstrap could be implemented by carefully selecting the subset size $m$ to properly inflate variance estimates \citep{wu2020evaluating}. 

Our methods rely on the assumption of no unmeasured confounding. Various causal methods have been developed to relax this assumption under specific conditions, such as instrumental variables, negative controls, and difference-in-differences \citep{athey2006identification, ertefaie2017tutorial,ye2023instrumented}. While most existing literature has focused on single time-point exposures, some approaches have been proposed to address time-varying unmeasured confounding \citep{chen2023estimating, cui2023instrumental,han2023optimal}. Extending methods for unmeasured confounding to structural nested models for recurrent event processes is an important direction for future research.

\section*{Data Availability}
Exposure data is publicly available at NASA's Socioeconomic Data and Applications Center (https://www.earthdata.nasa.gov/centers/sedac-daac). Census and ACS data are available from https://www.census.gov. Medicare data is provided by the US Centers for Medicare and Medicaid Services, however our data use agreement prohibits the sharing of the data sets used in our analysis. Academic and non-profit researchers interested in using Medicare data should contact the US Centers for Medicare and Medicaid Services directly.

\section*{Acknowledgments}
This work was supported by the National Institutes of Health grants R01ES034021, \linebreak R01ES035735, R01ES037156, R33NS120240, R01DA058996, and R01DA048764.
The computations in this paper were run on the FASSE cluster supported by the FAS Division of Science Research Computing Group at Harvard University and the Research Computing Environment supported by Institute for Quantitative Social Science in the Faculty of Arts and Sciences at Harvard University. 

\section*{Conflict of interest}
The authors have no competing interests to disclose.

\bibliographystyle{apalike}
\bibliography{ref2}

\clearpage
\appendix

\setcounter{section}{0}
\setcounter{equation}{0}

\renewcommand*{\thesection}{S\arabic{section}}
\renewcommand*{\thesubsection}{S\arabic{section}.\arabic{subsection}}
\renewcommand{\theequation}{S\arabic{equation}}
\renewcommand{\thefigure}{S\arabic{figure}}
\renewcommand{\thetable}{S\arabic{table}}
\renewcommand{\thetheorem}{S\arabic{theorem}}
\renewcommand{\thelemma}{S\arabic{lemma}}
\renewcommand{\bibnumfmt}[1]{[S#1]}
\renewcommand{\citenumfont}[1]{S#1}
\setcounter{page}{1}

{\centering\Large \textbf{Supplementary Materials}}

\section{Nuisance parameter estimation} \label{app:nuisance}
\subsection{Parametric estimation of $\mu_{km}(t)$} \label{app:prop}
Web Appendix A of \cite{seaman2020adjusting} shows that under the no unmeasured confounder assumption for any $t \geq k$, the conditional expectation of $A_k$ given $\bar A_{k-1}$, $\bar L_k$ and $D(\bar A_k,0) \geq t$ follows the same generalized linear model but with a 
time-dependent shifted intercept. That is, 
for any $t \geq k$,
\begin{align} \label{eq:prop}
g[  \mathbb{E}\{A_{k}|\oA_{k-1},\oL_k,D(\oA_{k},0) \geq t\} ] = 
g[ \mu_{k0}(t) ] = \tilde\theta_{k0}^\top H_k + \theta_k^\top \nu_k(t) ,    
\end{align}
where $g$ is an appropriate link function and $H_k=(\oA_{k-1},\oL_k)$. 

When $t \in [k,k+1)$ and $m=0$, $D(\oA_{k},0)$ corresponds to the observed event time $D$ and thus, the conditional expectation can be written as $g[  \mathbb{E}\{A_{k}|\oA_{k-1},\oL_k,D \geq t\} ] = \tilde\theta_{k0}^\top H_k + \theta_{k0} (t-k)$. To estimate $\theta_{k0}$, we first create a new  variable $\tilde t$ which consists of $\tilde m$ equally spaced values of $t \in [k,k+1) $. Then, for each individual, we create a copy with the same values of $(\oA_{k-1},\oL_k)$ and the new variable   $\tilde t$ (i.e., pseudo-data). Finally, we fit $g[  \mathbb{E}\{A_{k}|\oA_{k-1},\oL_k,\tilde t\} ] = \theta_{k0}^\top H_k + \theta_{kk} (\tilde t-k)$.

When $t \in [k,k+1)$ and $m=1$, the treatment model (\ref{eq:prop}) is $g[  \mathbb{E}\{A_{k-1}|\oA_{k-2},\oL_{k-1},\linebreak D(\oA_{k-1},0) \geq t\} ] = \tilde\theta_{k0}^\top H_k + \theta_{k0}+\theta_{k1} (\tilde t-k)$. The main challenge is that $D(\oA_{k-1},0)$ is no longer known. To overcome this issue, we first introduce  the following equality presented in Web Appendix B of \cite{seaman2020adjusting},
\begin{align}
    P \{D(\bar{A}_k, 0) \geq t \mid \bar{A}_k, \bar{L}_k, D(\bar{A}_k, 0) \geq k\} 
= E \{Y(t) w_{k1}(t) \mid \bar{A}_k, \bar{L}_k, D \geq k\},
\end{align}
where 
\begin{equation}
	w_{km}(t)= 
    \prod_{j=0}^{m-1}
    \exp\left\{
        A_{k-j}\boldsymbol\nu_{jk}(t)'\boldsymbol\alpha_m \right\}
\end{equation}
with $\boldsymbol\alpha_m=[\alpha_0,\ldots,\alpha_{m-1}]'$ is an $m$-length vector of parameters for the effect of exposure on the terminal event rate; and $\boldsymbol\nu_{jk}(t)$ is a vector of $j$ ones followed by $t-k$. 
Then, similar to the case with $m=0$, we fit a weighted model $g[  \mathbb{E}\{A_{k-1}|\oA_{k-2},\oL_{k-1},D(\oA_{k-1},0) \geq t\} ] = \tilde\theta_{(k-1)0}^\top H_{k-1} + \theta_{k0}+\theta_{k1} (\tilde t-k)$ using 
pseudo-data with weight $w_{k1}(\tilde t)= 
    \exp\left\{
        A_{k} \alpha_{0} (\tilde t -k) \right\}$. 
In Section \ref{app:risk} we discuss an estimation strategy for $\boldsymbol\alpha_m$. 

\subsection{Risk set adjusting weight model} \label{app:risk}

The parameter $\alpha_m$ represents the instantaneous (i.e., $m=0$) and delayed effect (i.e., $m>0$) of exposures on the probability of survival to time $t$. We can estimate  these parameters using the same estimating equations as those used for estimating  $\beta_m$ but replacing the recurrent event process with the survival counting process. Specifically, we define $\hat \alpha_0$ as a solution to 
\begin{equation}\label{eq:A0_est_eq}
    \sum_{i=1}^n\sum_{k=0}^K \int_{k}^{k+1}
    Y_i(t)
    \left\{A_{ik}-\mu_{ik0}(t)\right\}
    \left\{d\tilde N_i(t)-  A_{ik}\alpha_0dt\right\}
    =0,
\end{equation}
where $\mu_{k0}(t)=\mathbb{E}[A_{k}|\oA_{k-1},\oL_k,D(\oA_{k},0) \geq t]$ is a time-varying exposure model. Similarly, define $\hat \alpha_m$ as a solution to the following  unbiased estimating equation 
\begin{equation}\label{eq:Am_est_eq}
    \sum_{i=1}^n\sum_{k=0}^K\int_k^{k+1} 
    Y_i(t)w_{ikm}(t)
    \Delta_{ikm}(t)
    \left\{d\tilde N_i(t)- 
        \sum_{j=0}^{m-1}A_{i(k-j)}\alpha_j dt - 
        A_{i(k-m)}\alpha_m dt
    \right\}
    =0.
\end{equation}




\subsection{Estimating nuisance parameter, $d\rho_{km}$}
The general procedure of obtaining $d\hat \rho_{km}(t)$ is similar to the exposure model. Consider a case where $m=0$ and $t\in[k,k+1)$. First, we create a new variable $\tilde t$ which consists of $\tilde m$ equally spaced values of $t\in[k,k+1)$. Then for each individual create a copy with the same values of $(\oA_{k-1},\oL_k)$ and the new variable $\tilde t$ (i.e., pseudo-data). Define new variables $E^{m=0}_{k1}, E^{m=0}_{k2}, \ldots, E^{m=0}_{k\tilde m}$ that captures the total number of events during $\tilde m$ periods between $k$ to $k+1$. Specifically, for each individual $i$, define $E^{m=0}_{ik1}=\int_{k}^{k+1/\tilde m}dN_i(t)$, $E^{m=0}_{ik2}=\int_{k+1/\tilde m}^{k+2/\tilde m}dN_i(t)$, $\cdots$, $E^{m=0}_{ik\tilde m}=\int_{k+(\tilde m-1)/\tilde m}^{k+1}dN_i(t)$. Using an appropriate nonparametric statistical learning method (e.g., gradient boosting) for continuous outcomes we model $E^{m=0}_{kj}$ as a function of prior exposures, $\oA_{i(k-1)}$, covariates $\oL_{i(k)}$ and $(\tilde t -k)$ using the created pseudo-data.

For $m>0$ and $t\in[k,k+1)$ define $E^{m}_{ik1}=\int_{k}^{k+1/\tilde m}dN_i(t)-\sum_{j=0}^{m-1}A_{i(k-j)}\hat\beta_j$, $E^{m}_{ik2}=\linebreak \int_{k+1/\tilde m}^{k+2/\tilde m}dN_i(t)-\sum_{j=0}^{m-1}A_{i(k-j)}\hat\beta_j$, $\cdots$, $E^{m}_{ik\tilde m}=\int_{k+(\tilde m-1)/\tilde m}^{k+1}dN_i(t)-\sum_{j=0}^{m-1}A_{i(k-j)}\hat\beta_j$. These values represent the total number of events during $\tilde m$ periods $k-m$ to $k-m+1$ minus the `blip' due to exposures during $0$ to $m-1$ periods prior. We now model $E^{m}_{kj}$ as a function of prior exposures, $\oA_{i(k-m-1)}$, covariates $\oL_{i(k-m)}$ and $(\tilde t -k)$  using a created pseudo-data with weight $w_{km}(\tilde t)$.

\begin{remark}
    In the simulation and data analysis we fit a separate model for each combination of $k$ and $m$. In the simulation we allowed time-varying exposure and nuisance parameters setting $\tilde m=5$ (i.e., creating 5 pseudo-data with distinct variables $\tilde t$). Due to the large sample size in the data analysis, it is computationally infeasible to create pseudo-data and allow for a time-varying exposure model and nuisance parameter. However, it is realistic in the context of the problem to assume that each $\mu_{km}$ and $d\rho_{km}$ remains constant for any $t \in [l,l+1)$, $l=0,\cdots,K$. This implies that the functions are piecewise constant between each landmark time.
\end{remark}

\section{Asymptotics: parametric exposure model}

Suppose that our parametric nuisance functions satisfy  Assumption \ref{ass:parm}.

\subsection{Asymptotic linearity of $\hat \beta_0$ with parametric $\mu_0$}
\label{sec:par-proof-m.eq.0}
Under Assumption (\ref{ass:parm}a), we can write
\begin{align}
 \label{eq:mu-exp0}
 \hat\mu_{k0}(t)-\mu_{k0}(t)\ &=\ \dot\mu_{k0\tilde \theta_{k}}(t)(\hat\theta_{k}-\tilde\theta_{k})+o_P(n^{-1/2})\\ 
 \label{eq:mu-exp}
 &=(P_n-P_0)\dot\mu_{k0 \tilde\theta_{k}}(t) \phi_{k 0 \tilde\theta_{k}}+o_p(n^{-1/2}).
 \end{align}
Define the estimator $\hat \beta_0$ as
\[
\hat \beta_0 = \frac{P_n f(\hat \mu) }{P_n g(\hat \mu)},
\]
where $f(\hat \mu) = \sum_{k=0}^K \int_{k}^{k+1} Y(t)\left\{A_{k}-\hat \mu_{k0}(t)\right\}dN(t)$ and $g(\hat \mu) = \sum_{k=0}^K \int_{k}^{k+1} Y(t)\left\{A_{k}-\hat \mu_{k0}(t)\right\}A_{k}dt.$ 
Also, let $f(\mu)$ and
$g(\mu)$ denote these same
quantities with $\hat \mu_{k0}(t)$
replaced by $\mu_{k0}(t)$ for $k=0,\ldots,K.$

Letting $dN_{0}(t) = dN(t)-\tilde \beta_0 A_{k}dt,$ we then have
\begin{align*}
    \hat \beta_0 &= \frac{P_n \sum_{k=0}^K \int_{k}^{k+1} Y(t)\left\{A_{k}-\hat\mu_{k0}(t)\right\}dN(t)}{P_n g(\hat \mu)}\\
    &=\frac{P_n \sum_{k=0}^K \int_{k}^{k+1} Y(t)\left\{A_{k}-\hat\mu_{k0}(t)\right\}dN_{0}(t)}{P_n g(\hat \mu)}+\frac{P_n \sum_{k=0}^K \int_{k}^{k+1} Y(t)\left\{A_{k}-\hat\mu_{k0}(t)\right\}\beta_0A_{k}dt}{P_n g(\hat \mu)}\\
    & = \beta_0 +\frac{P_n \sum_{k=0}^K \int_{k}^{k+1} Y(t)\left\{A_{k}-\hat\mu_{k0}(t)\right\}dN_{0}(t)}{P_n g(\hat \mu)},
    \end{align*}
or that
\begin{equation}
\label{eq:betahat-par}
    \hat \beta_0 - \beta_0
= \frac{P_n \sum_{k=0}^K \int_{k}^{k+1} Y(t)\left\{A_{k}-\hat\mu_{k0}(t)\right\}dN_{0}(t)}{P_n g(\hat \mu)}.
\end{equation}
The no-unmeasured confounder assumption implies that the counterfactual counting process $N_{0}(t) \equiv N_{\bar A_{k-1},0}(t)$ is conditionally independent of $A_k$ given $\bar A_{k-1}$, $L_k$ and $D(\bar A_{k-1},0) \wedge C(\bar A_{k-1},0) \geq k$. Therefore,
\[
P_0 \sum_{k=0}^K \int_{k}^{k+1} Y(t)\left\{A_{k} - \mu_{k0}(t)\right\}dN_{0}(t) = 0.
\]

Consequently, we can write the numerator of (\ref{eq:betahat-par}) as
\begin{align*}
P_n& \sum_{k=0}^K  \int_{k}^{k+1} Y(t)  \left\{A_{k}-\hat\mu_{k0}(t)\right\}dN_{0}(t) 
  = \\
&(P_n-P_0)
\sum_{k=0}^K \int_{k}^{k+1} Y(t)\left\{A_{k} - \mu_{k0}(t)\right\}dN_{0}(t)
+
P_n 
\sum_{k=0}^K \int_{k}^{k+1} Y(t)\left\{\mu_{k0}(t)-\hat\mu_{k0}(t)\right\}dN_{0}(t).
\end{align*}
Using \eqref{eq:mu-exp0}, the second term on the right-hand side can be written
\begin{align*}
P_n& 
\sum_{k=0}^K \int_{k}^{k+1} Y(t)\left\{\mu_{k0}(t)-\hat\mu_{k0}(t)\right\}dN_{0}(t)
  =\\ 
 &-
\sum_{k=0}^K 
P_n \left[ 
\int_{k}^{k+1}
\dot\mu_{k0\tilde \theta_{k}}(t)
Y(t) d N_0(t) \right]
(\hat\theta_{k}-\tilde\theta_{k}) + o_p(n^{-1/2});
\end{align*}
defining
\[
\gamma_{k0} = 
P_0 \left[ 
\int_{k}^{k+1}
\dot\mu_{k0 \tilde\theta_{k}}(t)
Y(t) d N_0(t) \right]
\]
and using \eqref{eq:mu-exp} and the fact that
\[
(P_n - P_0) \left[ 
\int_{k}^{k+1}
\dot\mu_{k0\tilde \theta_{k}}(t)
Y(t) d N_0(t) \right]
(\hat\theta_{k }-\tilde\theta_{k })  = o_p(n^{-1/2}),
\]
we can now write
\[
P_n \sum_{k=0}^K \int_{k}^{k+1} Y(t) \left\{A_{k}-\hat\mu_{k0}(t)\right\}dN_{0}(t)
= G_n + o_p(n^{-1/2})
\]
where
\begin{align*}
G_n & = 
(P_n-P_0)
\sum_{k=0}^K \int_{k}^{k+1} Y(t)\left\{A_{k} - \mu_{k0}(t)\right\}dN_{0}(t)
-
(P_n-P_0)
\sum_{k=0}^K 
\phi_{k 0 \tilde\theta_{k}} \gamma_{k0}.
\end{align*}
Turning to the denominator of \eqref{eq:betahat-par}, note that
 \begin{align*}
     \frac{1 }{P_n g(\hat \mu)}-\frac{1 }{P_0 g( \mu)} &= \frac{P_n(g( \mu)-g(\hat \mu)) }{P_n g(\hat \mu)P_n g( \mu)}-\frac{(P_n-P_0)g(\mu) }{P_0 g( \mu)P_n g( \mu)}\\
     &=\frac{P_n(g( \mu)-g(\hat \mu)) }{P_n g(\hat \mu)P_n g( \mu)} - 
     \left[ \frac{(P_n-P_0)g(\mu) }{(P_0 g( \mu))^2} - \frac{\{(P_n-P_0)g(\mu)\}^2 }{(P_0 g( \mu))^2P_n g( \mu)} \right]\\
     &=\frac{P_n(g( \mu)-g(\hat \mu)) }{P_n g(\hat \mu)P_n g( \mu)} - \frac{(P_n-P_0)g(\mu) }{(P_0 g( \mu))^2} + o_p(n^{-1/2})\\
     &=\frac{P_0(g( \mu)-g(\hat \mu)) }{(P_0 g( \mu))^2} - \frac{(P_n-P_0)g(\mu) }{(P_0 g( \mu))^2} + o_p(n^{-1/2}).
 \end{align*}
 Since the first two terms on the right-hand side are
 $O_p(n^{-1/2}),$ the influence function is given by
\begin{align*}
    \hat \beta_0 -  \tilde\beta_0 &= \frac{(P_n-P_0) \sum_{k=0}^K  \int_{k}^{k+1} Y(t)\left\{A_{k}-\mu_{k0}(t)\right\}dN_{0}(t)}{P_0 g( \mu)} \\
    &\hspace{0.1in}-\frac{(P_n-P_0) \sum_{k=0}^K \phi_{k 0 \tilde\theta_{k}} \gamma_{k0}}{P_0 g( \mu)}+o_p(n^{-1/2}).
\end{align*}
That is, 
$\hat \beta_0 - \tilde \beta_0
= (P_n-P_0) \phi_{\tilde \beta_0}
+ o_P(n^{-1/2})$, where
\[
\phi_{\tilde \beta_0}
= \frac{1}{P_0 g(\mu)}
\times
\left\{
\sum_{k=0}^K  
\left[
\int_{k}^{k+1} Y(t)\left\{A_{k}-\mu_{k0}(t)\right\}dN_{0}(t)
-
 \phi_{k 0 \tilde\theta_{k}} \gamma_{k0}
\right] 
\right\}.
\]

\subsection{Asymptotic linearity of $\hat{\beta}_m$ with parametric $\mu_{km}$}
\label{sec:ALpar-km}

Following notational conventions similar
to Section \ref{sec:par-proof-m.eq.0},
for a given $m \geq 1$ define
\[
f(\mu,w,\beta_{0:(m-1)})=\sum_{k=0}^{K}\int_{k}^{k+1} Y(t) w_{km}(t) \{A_{k-m} - \mu_{km}(t) \}
     \{dN(t)-\sum_{j=0}^{m-1} A_{k-j}\beta_{j} dt\}
\]
and 
$g(\mu,w)=\sum_{k=0}^{K}\int_{k}^{k+1} Y(t){w}_{km}(t) \{A_{k-m} - \mu_{km}(t) \}A_{k-m}dt.$ We can then
define $\hat{\beta}_m$ as
\begin{align}
\hat{\beta}_m=\frac{P_nf(\hat{\mu},\hat{w},\hat{\beta}_{0:(m-1)})}{P_ng(\hat{\mu},\hat{w})},
\end{align}
where $\hat{w}_{km}(t)$ is as defined previously 
using estimates $\hat{\alpha}_0,\ldots,\hat{\alpha}_{m-1}$.
Under Assumption (\ref{ass:parm}.b),
\begin{align*}
 \hat\mu_{km}(t)-\mu_{km}(t)\ &=\ \dot\mu_{km\tilde \theta}(t)(\hat\theta_{km}-\tilde\theta_{km})+o_P(n^{-1/2})
 =(P_n-P_0)\dot\mu_{km \tilde\theta}(t) \phi_{km \tilde\theta}+o_P(n^{-1/2}).
 \end{align*}
and, in addition, 
 \begin{align*}
 \hat{w}_{km}(t)- w_{km}(t)\ &=\ \dot w_{km \alpha}(t)(\hat\alpha_{k}-\tilde\alpha_{k})+o_P(n^{-1/2})
 =(P_n-P_0)\dot w_{km \tilde\alpha}(t) \phi_{km \tilde\alpha}+o_P(n^{-1/2}).  
 \end{align*}

Importantly, the estimating equation leading to 
$\hat \beta_m$ for any valid choice of $m \geq 1$  depends on the prior sequence of
$\hat \beta_j, j = 0,\ldots,m-1.$ The influence
function for estimating $\tilde \beta_0$
is derived in Section \ref{sec:par-proof-m.eq.0}; that is, we obtain the result that
$\hat{\beta}_{0}-\tilde\beta_{0}=(P_n-P_0)\phi_{\tilde\beta_0} + o_p(n^{-1/2}).$
The influence
function derived 
for $\hat{\beta}_1  - \tilde\beta_1$
will depend on $\phi_{\tilde\beta_0};$
call this influence function $\phi_{ \tilde\beta_1};$ proceeding sequentially,
the influence function for
$\hat{\beta}_m  - \tilde\beta_m$
will depend on
$\phi_{ \tilde\beta_j}, j = 0,\ldots,
m-1.$  Hence, our proof for a given $m \geq 1$ will essentially
proceed by induction, where we will
assume that
$\hat{\beta}_{j}-\tilde\beta_{j}=(P_n-P_0)\phi_{\tilde\beta_j} + o_p(n^{-1/2})$ for each $j \leq m-1,$ and subsequently use these
to derive that for $\hat \beta_m.$

Fix $m \geq 1.$ Let $dN^{(m)}_0(t)=dN(t)-(\sum_{j=0}^{m-1} A_{k-j}\tilde\beta_{j}+A_{k-m}\tilde\beta_m)dt$. Then, note that
 \begin{align*}
    &\hat{\beta}_m = \frac{P_n\sum_{k=0}^{K}\int_{k}^{k+1}Y(t)\hat{w}_{km}(t)\{A_{k-m}-\hat{\mu}_{km}(t)\}\{dN(t)-(\sum_{j=0}^{m-1}A_{k-j} \hat{\beta}_{j})dt\}}{P_ng(\hat{\mu},\hat w)}\\
    &=\tilde\beta_m+ \frac{P_n\sum_{k=0}^{K}\int_{k}^{k+1}Y(t)\hat{w}_{km}(t)\{A_{k-m}-\hat{\mu}_{km}(t)\}\{dN(t)-(\sum_{j=0}^{m-1}A_{k-j} \hat{\beta}_{j}+A_{k-m}\tilde\beta_m)dt\}}{P_ng(\hat{\mu},\hat w)}.
\end{align*}

Recalling the definition 
$\Delta_{km}(t) = A_{k-m} - \mu_{km}(t),$
the numerator of the right hand side can be written as
\begin{align}
P_n\sum_{k=0}^{K}&\int_{k}^{k+1}Y(t)\hat{w}_{km}(t) \hat \Delta_{km}(t)
    \left\{dN(t)-\left(\sum_{j=0}^{m-1}A_{k-j} {\tilde\beta}_{j}-A_{k-m}\tilde\beta_m\right)dt\right\} - \label{eq:m1} \\
&P_n\sum_{k=0}^{K}\int_{k}^{k+1}Y(t)\hat{w}_{km}(t)
    \hat \Delta_{km}(t)
    \left\{\sum_{j=0}^{m-1}(\hat\beta_{j}-\tilde{\beta}_{j})A_{k-j}dt\right\}. \label{eq:m2}
\end{align}
The term \eqref{eq:m1} can be rewritten as follows:
\begin{align}
    (\ref{eq:m1}) = P_n\sum_{k=0}^{K}\int_{k}^{k+1}Y(t)\hat {w}_{km}(t)
    \hat \Delta_{km}(t)
    dN^{(m)}_0(t) 
     = (A) + (B) + (C) + (D),
 \end{align}
 where
 \[
 (A) = P_n \sum_{k=0}^{K}\int_{k}^{k+1}Y(t){w}_{km}(t) \Delta_{km}(t)
 dN^{(m)}_0(t),
 \]
 \[
 (B) = P_n \sum_{k=0}^{K}\int_{k}^{k+1}Y(t) {w}_{km}(t)\{{\mu}_{km}(t) - \hat{\mu}_{km}(t)\} dN^{(m)}_0(t),
 \]
 \[
 (C) = P_n \sum_{k=0}^{K}\int_{k}^{k+1}Y(t)\{ \hat {w}_{km}(t) -  {w}_{km}(t)\}
 \hat \Delta_{km}(t)
 dN^{(m)}_0(t),
 \]
 and
 \[
 (D) = P_n \sum_{k=0}^{K}\int_{k}^{k+1}Y(t)\{ \hat {w}_{km}(t) -  {w}_{km}(t)\} \{{\mu}_{km}(t) - \hat{\mu}_{km}(t)\} dN^{(m)}_0(t).
 \]
 Under our assumptions, we have 
 \[
 P_0 \sum_{k=0}^{K}\int_{k}^{k+1}Y(t){w}_{km}(t)
 \Delta_{km}(t)
 dN^{(m)}_0(t) = 0  
 \]
and $(D) = o_p(n^{-1/2});$ hence, $(\ref{eq:m1}) = (A') + (B) + (C) + o_P(n^{-1/2}),$ where
 \[
 (A') = (P_n-P_0) \sum_{k=0}^{K}\int_{k}^{k+1}Y(t){w}_{km}(t)
 \Delta_{km}(t)
 dN^{(m)}_0(t).
\]
Consider term (B); using expansions given earlier,
\begin{align}
(B) &= P_n \sum_{k=0}^{K}\int_{k}^{k+1}Y(t) {w}_{km}(t)\{{\mu}_{km}(t) - \hat{\mu}_{km}(t)\} dN^{(m)}_0(t)  \nonumber \\
& = - P_n \sum_{k=0}^{K}\int_{k}^{k+1}Y(t) {w}_{km}(t)\{ \dot\mu_{km\tilde \theta}(t)(\hat\theta_{k}-\tilde\theta_{k})+o_P(n^{-1/2}) \} dN^{(m)}_0(t)  \nonumber \\
& = -  \sum_{k=0}^{K}  (\hat\theta_{k}-\tilde\theta_{k})  \, P_n \int_{k}^{k+1}Y(t) {w}_{km}(t) \dot\mu_{km\tilde \theta_{}}(t) dN^{(m)}_0(t)   +o_P(n^{-1/2}) \nonumber \\
& = - \sum_{k=0}^{K} (P_n-P_0) \phi_{km \tilde \theta}  \, P_n \int_{k}^{k+1}Y(t) {w}_{km}(t) \dot\mu_{km\tilde \theta_{}}(t) dN^{(m)}_0(t)   +o_P(n^{-1/2}) \nonumber \\
& = - (P_n-P_0) \sum_{k=0}^{K}  \phi_{km \tilde \theta}  \, P_0 \int_{k}^{k+1}Y(t) {w}_{km}(t) \dot\mu_{km\tilde \theta_{}}(t) dN^{(m)}_0(t)   +o_P(n^{-1/2}) \nonumber \\ 
& = - (P_n-P_0) \sum_{k=0}^{K}  
\phi_{km \tilde \theta}  \,
\xi^{(1)}_{km \tilde \theta_k} 
+o_P(n^{-1/2}) \nonumber 
\end{align}
where
$\xi^{(1)}_{km \tilde \theta_k} = P_0 \int_{k}^{k+1}Y(t) {w}_{km}(t) \dot\mu_{km\tilde \theta_{}}(t) dN^{(m)}_0(t).$
A similar series of arguments shows that term (C) can be written
\begin{align}
(C) &= \sum_{k=0}^{K}   (\hat\alpha_{k}-\tilde\alpha_{k}) 
 \, P_n \int_{k}^{k+1}Y(t) \dot w_{km \tilde \alpha}(t)  
 \Delta_{km}(t)
 dN^{(m)}_0(t)   +o_P(n^{-1/2}). \nonumber \\
 & = 
 (P_n-P_0) \sum_{k=0}^{K-m-1}   
 \phi_{km\tilde \alpha} 
 \, P_0 \int_{k}^{k+1}Y(t) \dot w_{km \tilde \alpha}(t)  
 \Delta_{km}(t)
 dN^{(m)}_0(t)   +o_P(n^{-1/2}) \nonumber \\
& = 
 (P_n-P_0) \sum_{k=0}^{K-m-1}   
 \phi_{km \tilde \alpha} 
 \, \xi^{(2)}_{km \tilde \alpha} +o_P(n^{-1/2}). \nonumber 
 \end{align}
 where
 $\xi^{(2)}_{km \tilde \alpha} = P_0 \int_{k}^{k+1}Y(t) \dot w_{km \tilde \alpha}(t)  \Delta_{km}(t)
 dN^{(m)}_0(t).$

Now, recall that we have assumed
that $\hat{\beta}_{j}$ has
influence function 
$\phi_{\tilde\beta_j}$
for $j \leq m-1.$
Similar arguments to those given 
above then show that the  term (\ref{eq:m2}) can be linearized as follows:
\begin{align}
\label{eq: m2simp}
    (\ref{eq:m2}) &= 
(P_n-P_0) \left\{ 
\sum_{j=0}^{m-1} \phi_{\tilde\beta_j} \sum_{k=0}^{K}  
    \int_{k}^{k+1}  \gamma_{kmj}(t) dt 
    \right\} + o_p(n^{-1/2}),
\end{align}
where $\gamma_{kmj}(t) = P_0  Y(t) w_{km}(t)
\Delta_{km}(t) 
A_{k-j}
=
P_0  Y(t) w_{km}(t)
\{A_{k-m}-{\mu}_{km}(t)\}
A_{k-j}.$
Consequently, collecting terms,
\begin{align*}
\hat{\beta}_m & - \tilde\beta_m
= \frac{1}{P_0 g(\mu,w)} \times
\biggl[
 (P_n-P_0) \sum_{k=0}^{K}\int_{k}^{k+1}Y(t){w}_{km}(t)\{A_{k-m}-{\mu}_{km}(t)\} dN^{(m)}_0(t) \\
& + (P_n-P_0) \sum_{k=0}^{K} 
\{ \phi_{km \tilde \alpha} 
 \, \xi^{(2)}_{km \tilde \alpha}
- \phi_{km \tilde \theta}  \,
\xi^{(1)}_{km \tilde \theta_k} \}
\\
& -
(P_n-P_0) \left\{ 
\sum_{j=0}^{m-1} \phi_{\tilde\beta_j} \sum_{k=0}^{K}  
    \int_{k}^{k+1}  \gamma_{kmj}(t) dt 
    \right\} 
\biggr]
+  o_P(n^{-1/2});
\end{align*}
or, equivalently,
\begin{align*}
\hat{\beta}_m  - \tilde\beta_m
= \frac{1}{P_0 g(\mu_0,w)} & \times
(P_n-P_0) \sum_{k=0}^{K} \biggl[
 \int_{k}^{k+1}Y(t){w}_{km}(t)\{A_{k-m}-{\mu}_{km}(t)\} dN^{(m)}_0(t) \\
& + 
\{ \phi_{km \tilde \alpha} 
 \, \xi^{(2)}_{km \tilde \alpha}
- \phi_{km \tilde \theta}  \,
\xi^{(1)}_{km \tilde \theta_k} \}
-
 \sum_{j=0}^{m-1} \phi_{\tilde\beta_j} 
    \int_{k}^{k+1}  \gamma_{kmj}(t) dt 
\biggr]
+  o_P(n^{-1/2}).
\end{align*}

\section{Asymptotics: nonparametric exposure model}

In this section, we proceed similarly
to the parametric case,
first establishing the result
for $m = 0$
with nonparametric $\mu_{k0}$ and $d \rho_{k0},$
and then using
this result to sequentially 
establish the asymptotic linearity of $\hat \beta_m.$

\def\E{E}
\def\R{\mathbbmss{R}}

Below, we give a useful lemma that will assist in proving
the required results; this is a generalization of Lemma 4 in
\cite{ErtefaieRQ}. As preparatory notation, 
let $\| \bm{x} \|_q$ denote the usual 
$q-$ norm of a vector $\bm{x}$  for $q=1,2,\infty$.

\begin{lemma}\label{lem:help2cvg}
Let $\bm{B}_1,\ldots, \bm{B}_N$ be independent, identically distributed vectors from $P_0,$ where $\bm{B}_i \in {\mathcal B} \subset \R^d.$
Let $\bm I_n$ be a randomly chosen subset of the integers $1,\ldots,N$ of length $n = O(N)$ and let its complement 
$\bm I^c_n$ have $N-n$ elements.
Let $\bm{F}_{I_n}$ and $\bm{F}_{I^c_n}$ be the corresponding disjoint subsets of $\bm{B}_1,\ldots,\bm{B}_{N}.$ 
Let $\gamma: {\mathcal B} \rightarrow \R$ and let $\hat \gamma(\cdot; \bm{F}_{I^c_n})$ be an estimator of $\gamma(\cdot)$ 
derived from the data $\bm{F}_{I^c_n}.$  Finally, let
$D_i \subset B_i$ be a strictly smaller subset of the data
in $B_i$, and let
$h(B_i; \hat \gamma(\bm{B}_i; \bm{F}_{I^c_n}))$ be a 
$d \times 1$ vector function
of the data $B_i \in I_n$ such that
$ E\left( h_j(B_i; \hat \gamma(\bm{B}_i; \bm{F}_{I^c_n})) \mid D_i, \bm{F}_{I^c_n} \right) = 0$ and
$\mbox{Var}\left( h_j(B_i; \hat \gamma(\bm{B}_i; \bm{F}_{I^c_n})) \mid D_i, \bm{F}_{I^c_n} \right) = o_P(N^{-a})$ for some $a \geq 0.$ Then,
\begin{equation*}
\bm{L}_{n,N}  =  \frac{1}{n} \sum_{i \in \bm{I}_n} 
h(B_i; \hat \gamma(\bm{B}_i; \bm{F}_{I^c_n}))
\end{equation*}
satisfies
$\| \bm{L}_{n,N} \|_{\infty} = O_p(N^{-(1+a)/2}).$
\end{lemma}

\begin{proof}
The proof  relies on a variant of Chebyshev's inequality. 
Let ${\mathcal D}_n = \{ \bm I_n,  ( \bm{D}_k, k \in \bm I_n) \}.$ Then,
for each component $j = 1,\ldots,d$ it is easy to see that
\[
\E\left( \bm{L}_{n,N,j} \big| \bm{F}_{I^c_n} \right) 
~=~
\E \left\{ 
\E\left( \bm{L}_{n,N,j} \big| \bm{F}_{I^c_n},  {\mathcal D}_n \right) 
\big| \bm{F}_{I^c_n} \right\}
~=~0;
\]
this follows from calculating the inner expectation on the right-hand-side and using
the assumption that 
$ E\left( h_j(B_i; \hat \gamma(\bm{B}_i; \bm{F}_{I^c_n})) \mid D_i, \bm{F}_{I^c_n} \right) = 0$.
Using a similar conditioning argument, 
\begin{eqnarray*}
\mbox{var}\left( \bm{L}_{n,N,j} \big| \bm{F}_{I^c_n}\right) & = & 
\E\left\{ \mbox{var}\left( \bm{L}_{n,N,j} \big| \bm{F}_{I^c_n},  {\mathcal D}_n \right) \big| \bm{F}_{I^c_n} \right\}.
\end{eqnarray*}
Straightforward calculations give
\begin{eqnarray*}
\mbox{Var}\left( \bm{L}_{n,N,j} \big| \bm{F}_{I^c_n},  {\mathcal D}_n \right)
& = & 
 \frac{1}{n^2} \sum_{i \in \bm{I}_n}
 \mbox{Var}\left( h_j(B_i; \hat \gamma(\bm{B}_i; \bm{F}_{I^c_n})) \mid D_i, \bm{F}_{I^c_n} \right) 
 \end{eqnarray*}
implying that
\[
\mbox{Var}\left( \bm{L}_{n,N,j} \big| \bm{F}_{I^c_n} \right) 
 =  n^{-1} o_P(N^{-a}) \\
  =  o_p(N^{-(1+a)}),
\]
the last step following from the assumptions of the Lemma
and the fact that $n = O(N)$.  
Using a vector form of Chebyshev's inequality,
it can then be shown that $\| \bm{L}_{n,N} \|_2 = o_p(N^{-(1+a)/2});$ since $\| \bm{L}_{n,N} \|_{\infty} \leq \| \bm{L}_{n,N} \|_2$, the stated result follows.
\end{proof}

Under the conditions of the previous lemma, and upon inspection
of the proof, the following condition is also sufficient for the stated results to hold: for each $j = 1,\ldots,d$ and every $n$ sufficiently large,
\[
\mbox{Var}\left( h_j(B_i; \hat \gamma(\bm{B}_i; \bm{F}_{I^c_n})) \right)
=
E\left[ \mbox{Var}\left( h_j(B_i; \hat \gamma(\bm{B}_i; \bm{F}_{I^c_n})) \mid  \bm{F}_{I^c_n} \right) \right] = o(N^{-a}), ~~ i = 1,\ldots,n
\]
for some $a \geq 0.$ For, under this condition and the
other conditions in the lemma, it follows
that
$
\mbox{var}\left( \bm{L}_{n,N,j}  \right)= o(N^{-(1+a)}).$
We close this section by stating an assumption that will be used to prove 
Theorem \ref{thm:nonparmodels}.

\begin{assumption}[Nonparametric models for $\mu_{km}$ and $d\rho_{km},$ $m \geq 0.$] \label{ass:nonparm-new}
Let
${\cal S}^1_{n,v}$ be the training sample
that defines $P_{n,v}^1,$ with ${\cal S}^0_{n,v}$ denoting its complement. 
In addition, define
$w_{k0}(t) = 1$ and
$w_{km}(t)$ as 
in \eqref{eq:weight_def},
and set $dM^{(m)}_0(t)=dN(t)- d \rho_{km}(t) - (\sum_{j=0}^{m-1} A_{k-j}\tilde\beta_{j}+\Delta_{km}(t) \tilde\beta_m)dt.$
Then, when the nuisance parameters are estimated nonparametrically
using the data in ${\cal S}^0_{n,v}$,
we impose the following rate conditions for each $k = 0,\ldots,
K-1$ and $m \geq 0:$
\begin{align*}
 E\biggl[ \mbox{Var}\biggl\{ 
 \sum_{k=0}^K \int_{k}^{k+1} Y(t)
w_{km}(t) \Delta_{km}(t) \{d\hat\rho_{kmv}(t)-d\rho_{km}(t)\} \mid 
{\cal S}^{0}_{n,v} \biggr\} \biggr] & = o(1);\\
E\biggl[ \mbox{Var}\biggl\{ \sum_{k=0}^K  \int_{k}^{k+1} Y(t)
w_{km}(t) \Delta_{km}(t) \{\hat\mu_{kmv}(t)-\mu_{km}(t)\} dt \mid 
{\cal S}^{0}_{n,v} \biggr\} \biggr] & =  o(1);\\
E\biggl[ \mbox{Var}\biggl\{ \sum_{k=0}^K  \int_{k}^{k+1} Y(t)
w_{km}(t) \{\hat\mu_{kmv}(t)-\mu_{km}(t)\}
d M^{(m)}_0(t)
\mid  
{\cal S}^{0}_{n,v} \biggr\}  \biggr] & =  o(1);\\
\mbox{E}\biggl\{ \sum_{k=0}^K \int_{k}^{k+1} 
 Y(t) w_{km}(t) \{\hat\mu_{kmv}(t)-\mu_{km}(t)\} \{
d\hat\rho_{kmv}(t) -  d\rho_{km}(t) \}
  \biggr\}
& = o(n^{-1/2}) \\
\mbox{E}\biggl\{ \sum_{k=0}^K \int_{k}^{k+1}
Y(t) w_{km}(t)
\{\hat\mu_{kmv}(t)-\mu_{km}(t)\}^2 dt
\biggr\}
& = o(n^{-1/2}). 
\end{align*}
\end{assumption}

\subsection{Asymptotic linearity of $\hat \beta_0$ with nonparametric $\mu_{k0},$
$d \rho_{k0}$}
\label{sec:ALnonpar-k0}

\noindent
The following proof imposes Assumption \ref{ass:nonparm-new}
to prove Theorem \ref{thm:nonparmodels} when $m=0.$ 
Let \[
f(\mu,d\rho) = \sum_{k=0}^K \int_{k}^{k+1} Y(t)\left\{A_{k}-\mu_{k0}(t)\right\}\{dN(t) - d\rho_{k0}(t)\} 
\]
and $g(\mu) = \sum_{k=0}^K \int_{k}^{k+1} Y(t)\left\{A_{k}-\mu_{k0}(t)\right\}^2dt$.  Define the estimator $\hat \beta_0$ as
\[
\hat \beta_0 = \frac{\sum_{v=1}^V P_{n,v}^1 f(\hat \mu_v,d \hat \rho_v) }{\sum_{v=1}^V P_{n,v}^1 g(\hat \mu_v)}, 
\]
where $\hat \mu_v$ and $d\hat \rho_v$ are estimates of $\mu$ and $d\rho$ obtained from the $v^{th}$ training sample. 
Defining $\hat \Delta_{k0v}(t) = A_{k}-\mu_{k0v}(t),$
note also that 
\[
\sum_{v=1}^V P_{n,v}^1 g(\hat \mu_v)
= \sum_{v=1}^V P_{n,v}^1 \sum_{k=0}^K \int_{k}^{k+1} Y(t) \hat \Delta^2_{k0v}(t)  dt
\]
and hence that
\[
\tilde \beta_0 = \frac{\sum_{v=1}^V P_{n,v}^1 \sum_{k=0}^K \int_{k}^{k+1} Y(t) \hat \Delta^2_{k0v}(t) 
\tilde \beta_0dt}{\sum_{v=1}^V P_{n,v}^1 g(\hat \mu_v)}.
\]
Therefore, we can write
\begin{align}
\nonumber
    \hat \beta_0 &= \frac{\sum_{v=1}^V P_{n,v}^1 \sum_{k=0}^K \int_{k}^{k+1} Y(t)\hat\Delta_{k0v}(t)\{dN(t) - 
    d\hat\rho_{k0v}(t)\}}{\sum_{v=1}^V P_{n,v}^1 g(\hat \mu_v)}\\
    \label{beta0-rep}
    & = \tilde\beta_0 +\frac{V^{-1}\sum_{v=1}^V P_{n,v}^1 \sum_{k=0}^K \int_{k}^{k+1} Y(t)\hat\Delta_{k0v}(t)\{dN(t) - d\hat\rho_{k0v}(t)-\hat\Delta_{k0v}(t)\tilde\beta_0 dt\}}{V^{-1}\sum_{v=1}^V P_{n,v}^1 g(\hat \mu_v)}.
\end{align}
Fixing $v,$ consider the numerator
in the fraction on the right-hand side.
We can decompose
\[
\hat {\cal H}_{n,v} = P_{n,v}^1 \sum_{k=0}^K \int_{k}^{k+1} Y(t)\hat\Delta_{k0v}(t)\{dN(t) - d\hat\rho_{k0v}(t)-\hat\Delta_{k0v}(t)\tilde\beta_0 dt\}
\]
into a sum of six terms:
$\hat {\cal H}_{n,v} = (I)-(II)+(III)
-(IV)+(V)-(VI)$ where
\begin{align*}
(I) & = P_{n,v}^1 \sum_{k=0}^K \int_{k}^{k+1} Y(t) \Delta_{k0}(t) 
\{dN(t) - d\rho_{k0}(t)-\Delta_{k0}(t)\tilde\beta_0 dt\} \\
(II) & = P_{n,v}^1 \sum_{k=0}^K \int_{k}^{k+1} Y(t)
\Delta_{k0}(t) \{d\hat\rho_{k0v}(t)-d\rho_{k0}(t)\}\\
(III) & = \tilde \beta_0  P_{n,v}^1 \sum_{k=0}^K \int_{k}^{k+1} Y(t)
\Delta_{k0}(t) \{\hat\mu_{k0v}(t)-\mu_{k0}(t)\} dt\\
(IV) & = P_{n,v}^1 \sum_{k=0}^K \int_{k}^{k+1} Y(t)
\{\hat\mu_{k0v}(t)-\mu_{k0}(t)\}
\{dN(t) - d\rho_{k0}(t)-\Delta_{k0}(t)\tilde\beta_0 dt\}\\
(V) & = P_{n,v}^1 \sum_{k=0}^K \int_{k}^{k+1} Y(t)
\{\hat\mu_{k0v}(t)-\mu_{k0}(t)\}
\{d\hat\rho_{k0v}(t)-d\rho_{k0}(t)\}\\
(VI) & = \tilde \beta_0 P_{n,v}^1 \sum_{k=0}^K \int_{k}^{k+1} Y(t)
\{\hat\mu_{k0v}(t)-\mu_{k0}(t)\}^2 dt.
\end{align*}
Under the structural modeling assumptions of this paper,
$E[\Delta_{k0}(t) | 
\bar A_{k-1}, \bar L_k, Y(t) = 1] = 0$
and that
$E[
dN(t) - d\rho_{k0}(t)-\Delta_{k0}(t)\tilde\beta_0 dt
| 
\bar A_{k-1}, \bar L_k, Y(t) = 1] = 0.$
Hence, it can be seen that each of the
terms in $(II) - (IV)$ are mean zero (i.e., given
the training sample ${\cal S}^0_{n,v}$). The relevant variance conditions
in Assumption \ref{ass:nonparm-new}, combined with
Lemma \ref{lem:help2cvg}, then imply that 
$(II) = (III) = (IV) = o_P(n^{-1/2}).$
In contrast, neither $(V)$ nor $(VI)$ have mean zero conditional
on the training sample. Hence, the proof that these terms 
are negligible requires a different argument. 
In particular, following a similar argument
to Lemma \ref{lem:help2cvg} that employs Markov's inequality, 
the relevant expectation condition
in Assumption \ref{ass:nonparm-new} 
implies that $(VI) = o_P(n^{-1/2});$ a similar
argument can be used for $(V)$.

In view of the results established above,
and using term $(I),$ the asymptotics
for cross-fit estimators lead to
the equivalence result
\[
V^{-1} \sum_{v=1}^V\hat {\cal H}_{n,v} 
= (P_{n} - P_0) \sum_{k=0}^K \int_{k}^{k+1} Y(t) \Delta_{k0}(t) 
\{dN(t) - d\rho_{k0}(t)-\Delta_{k0}(t)\tilde\beta_0 dt\}
+ o_P(n^{-1/2});
\]
hence, returning to \eqref{beta0-rep}, and 
using arguments analogous to 
Section \ref{sec:par-proof-m.eq.0},
we now have
\[
\hat \beta_0 -
\tilde\beta_0 =
(P_{n} - P_0)
\frac{\sum_{k=0}^K \int_{k}^{k+1} Y(t) \Delta_{k0}(t) 
\{dN(t) - d\rho_{k0}(t)-\Delta_{k0}(t)\tilde\beta_0 dt\}}{P_0 g(\mu)}
+ o_P(n^{-1/2}).
\]

\subsection{Asymptotic linearity of $\hat \beta_m$ with nonparametric $\mu_m$}

This proof
imposes Assumption \ref{ass:nonparm-new}
to prove Theorem \ref{thm:nonparmodels} when $m\geq 1.$ The proof for a given $m \geq 1$   proceeds by induction; we will
assume that $\hat{\beta}_{j}-\tilde\beta_{j}=(P_n-P_0) \phi^\star_{\tilde\beta_j} + o_p(n^{-1/2})$ for each $j \leq m-1,$ and subsequently use these
to derive that for $\hat \beta_m.$ The previous
section establishes the needed result for
$m=0.$ We note that the influence function
$\phi^\star_{\tilde\beta_j}$ will in general
differ from $\phi_{\tilde\beta_j}$
as defined in Section \ref{sec:ALpar-km}.

In general, fix $m \geq 1$ and  let 
$f(\mu,w,\beta_{0:(m-1)})=\sum_{k=0}^{K}\int_{k}^{k+1} Y(t) w_{k}(t) \Delta_{km}(t)
     \{dN(t)-  d\rho_{km}(\bar A_{k-m-1},\bar L_{k-m})-\sum_{j=0}^{m-1} A_{k-j}\beta_{j}dt\}$ 
and 
$g(\mu,w)=\sum_{k=0}^{K}\int_{k}^{k+1} Y(t){w}_{km}(t) \Delta_{km}^2(t)dt$ where $\Delta_{km}(t)=A_{k-m} - \mu_{km}(t)$. 
Define $\hat{\beta}_m$ as
\begin{align}
    \hat{\beta}_m=\frac{\sum_{v=1}^V P_{n,v}^1 f(\hat{\mu}_v,\hat{w},d\hat \rho_{kmv},\hat{\beta}_{0:(m-1)})}{\sum_{v=1}^V P_{n,v}^1 g(\hat{\mu}_v,\hat{w})},
\end{align}
where $\hat{w}_{km}(t)$ is defined using the full-sample
estimates $\hat{\alpha}_0,\ldots,\hat{\alpha}_{m-1}$. 
We begin by writing
\begin{align}
    \hat{\beta}_m &= \frac{V^{-1} \sum_{v=1}^V P_{n,v}^1\sum_{k=0}^{K}\int_{k}^{k+1}Y(t)\hat{w}_{km}(t)\hat\Delta_{kmv}(t)\{dN(t)- d\hat \rho_{kmv}(t)-\sum_{j=0}^{m-1}A_{k-j} \hat{\beta}_{j}dt\}}{V^{-1} \sum_{v=1}^V P_{n,v}^1g(\hat{\mu},\hat w)} \nonumber\\
    \label{eq:beta-nonpar}
    &= \tilde{\beta}_m+\left\{V^{-1} \sum_{v=1}^V P_{n,v}^1 g(\hat{\mu},\hat w)\right\}^{-1}\times 
    \left\{ V^{-1} \sum_{v=1}^V \hat {\cal W}_{n,v}
    \right\},
\end{align}
where 
\begin{align*}
\hat {\cal W}_{n,v} = 
P_{n,v}^1\sum_{k=0}^{K}\int_{k}^{k+1}Y(t)\hat{w}_{km}(t)\hat\Delta_{kmv}(t)\{dN(t)- d\hat \rho_{kmv}(t)-\sum_{j=0}^{m-1}A_{k-j} \hat{\beta}_{j}dt-\hat\Delta_{kmv}(t)\tilde\beta_mdt\}
\end{align*}
and
$\hat\Delta_{kmv}(t)=A_{k-m}-\hat{\mu}_{kmv}(t)$.

Under Assumption (\ref{ass:parm}.b), we recall that
\begin{align*}
 \hat{w}_{km}(t)- w_{km}(t)\ &=\ \dot w_{km \alpha}(t)(\hat\alpha_{k}-\tilde\alpha_{k})+o_P(n^{-1/2})
 =(P_n-P_0)\dot w_{km \tilde\alpha}(t) \phi_{km \tilde\alpha}+o_P(n^{-1/2});  
 \end{align*}
 in addition, we recall that
 $\hat{\beta}_{j}-\tilde\beta_{j}=(P_n-P_0) \phi^\star_{\tilde\beta_j} + o_p(n^{-1/2})$ for each $j \leq m-1.$
 Using the implications that
 $\hat{w}_{km}(t)- w_{km}(t) = O_P(n^{-1/2})$
 and $\hat{\beta}_{j}-\tilde\beta_{j}
 = O_P(n^{-1/2}),$ along with
 Assumption \ref{ass:nonparm-new},
 the term $\hat {\cal W}_{n,v}$
 can be decomposed into 
 a sum of six terms:
$\hat {\cal W}_{n,v} = (I)-(II)+(III)
+(IV)-(V)-(VI)$ where
\begin{align*}
(I) & = P_{n,v}^1 \sum_{k=0}^K \int_{k}^{k+1} 
Y(t)\hat{w}_{km}(t) \Delta_{km}(t)\{dN(t)- d\rho_{km}(t)-\sum_{j=0}^{m-1}A_{k-j} \hat{\beta}_{j}dt-\Delta_{km}(t)\tilde\beta_mdt\}
\\
(II) & = P_{n,v}^1 \sum_{k=0}^K \int_{k}^{k+1} Y(t)
w_{km}(t) \Delta_{km}(t) \{d\hat\rho_{kmv}(t)-d\rho_{km}(t)\} + o_P(n^{-1/2}) \\
(III) & = \tilde \beta_m  P_{n,v}^1 \sum_{k=0}^K \int_{k}^{k+1} Y(t)
w_{km}(t) \Delta_{km}(t) \{\hat\mu_{kmv}(t)-\mu_{km}(t)\} dt + o_P(n^{-1/2})\\
(IV) & = P_{n,v}^1 \sum_{k=0}^K \int_{k}^{k+1} Y(t)
w_{km}(t)
\{\hat\mu_{kmv}(t)-\mu_{km}(t)\}
\{d\hat\rho_{kmv}(t)-d\rho_{km}(t)\}
+ o_P(n^{-1/2})\\
(V) & = \tilde \beta_0 P_{n,v}^1 \sum_{k=0}^K \int_{k}^{k+1} Y(t) w_{km}(t)
\{\hat\mu_{kmv}(t)-\mu_{km}(t)\}^2 dt
+ o_P(n^{-1/2}) \\
(VI) & = P_{n,v}^1 \sum_{k=0}^K \int_{k}^{k+1} Y(t)
w_{km}(t) 
\{\hat\mu_{kmv}(t)-\mu_{km}(t)\}
dM^{(m)}_0(t) + o_P(n^{-1/2})
\end{align*}
The results in $(II)-(V)$
make use of the assertion that
$\hat{w}_{km}(t)- w_{km}(t) = O_P(n^{-1/2}),$
along with the implied result from Assumption 
\ref{ass:nonparm-new} that this difference
multiplies a term that is (at least) $o_P(1).$
The result in (VI) makes use of 
these same results, along with 
$\hat{\beta}_{j}-\tilde\beta_{j}
 = O_P(n^{-1/2}),$ and introduces
$d M^{(m)}_0(t)$ 
as defined in Assumption \ref{ass:nonparm-new}.

The relevant variance conditions
in Assumption \ref{ass:nonparm-new}, combined with
Lemma \ref{lem:help2cvg}, further imply 
that each of the leading terms in $(II) - (VI)$
are $o_P(n^{-1/2}).$ Hence, the large sample
behavior is determined by $(I)$. Using arguments
similar to the above, term $(I)$ can be written
as ${\cal Q}_{n,v} + o_P(n^{-1/2})$ where
\begin{align}
\label{I-1}
{\cal Q}_{n,v} = 
P_{n,v}^1 & \sum_{k=0}^K
\int_{k}^{k+1} 
Y(t)w_{km}(t) \Delta_{km}(t) dM^{(m)}_0(t) \\
\label{I-2}
& +
P_{n,v}^1 \sum_{k=0}^K
\int_{k}^{k+1} 
Y(t)(\hat w_{km}(t) - w_{km}(t)) \Delta_{km}(t) dM^{(m)}_0(t) \\
\label{I-3}
& -P_{n,v}^1 \sum_{k=0}^K
\int_{k}^{k+1} 
Y(t) \hat w_{km}(t) \Delta_{km}(t) 
\left[
\sum_{j=0}^{m-1}A_{k-j} (\hat{\beta}_{j} - \tilde{\beta}_j) 
\right] dt 
\end{align}
Using $
 \hat{w}_{km}(t)- w_{km}(t)\ =
 (P_n-P_0)\dot w_{km \tilde\alpha}(t) \phi_{km \tilde\alpha}+o_P(n^{-1/2}),$
 we can show that
\[
V^{-1} \sum_{v=1}^V
\mbox{\eqref{I-2}}
=
(P_n - P_0) \sum_{k=0}^K \phi_{km \tilde\alpha} \,
\xi^{(3)}_{km\tilde\alpha} + o_P(n^{-1/2}),
\]
where
$\xi^{(3)}_{km\tilde\alpha}  =
P_0
\int_{k}^{k+1} 
Y(t) \dot w_{km \tilde\alpha}(t)  \Delta_{km}(t) dM^{(m)}_0(t).$
Similarly, using 
$\hat{\beta}_{j}-\tilde\beta_{j}=(P_n-P_0) \phi^\star_{\tilde\beta_j} + o_p(n^{-1/2})$ for each $j \leq m-1$ and using arguments
identical to those used simplify
\eqref{eq:m2}, we have
\begin{align*}
   V^{-1} \sum_{v=1}^V (\ref{I-3}) &= 
(P_n-P_0) \left\{ 
\sum_{j=0}^{m-1} \phi^{\star}_{\tilde\beta_j} \sum_{k=0}^{K}  
    \int_{k}^{k+1}  \gamma_{kmj}(t) dt 
    \right\} + o_p(n^{-1/2}),
\end{align*}
where $\gamma_{kmj}(t) = P_0  Y(t) w_{km}(t)
\Delta_{km}(t) 
A_{k-j}
=
P_0  Y(t) w_{km}(t)
\{A_{k-m}-{\mu}_{km}(t)\}
A_{k-j}$ is the same 
as before. Arguing similarly
to the case where $m=0,$
it follows that
\begin{align*}
V^{-1} \sum_{v=1}^V
{\cal Q}_{n,v}
= (P_n-P_0) 
\biggl[ &
\sum_{k=0}^K \biggl\{ 
\int_{k}^{k+1} 
Y(t)w_{km}(t) \Delta_{km}(t) dM^{(m)}_0(t) \\
&
+
 \phi_{km \tilde\alpha} \,
\xi^{(3)}_{km\tilde\alpha}
-
\sum_{j=0}^{m-1} \phi^{\star}_{\tilde\beta_j}   
    \int_{k}^{k+1}  \gamma_{kmj}(t) dt
    \biggr\}
    \biggr]+ o_p(n^{-1/2}).
\end{align*}
Returning to \eqref{eq:beta-nonpar},
putting everything together and
using arguments analogous to 
Section \ref{sec:ALpar-km}
now leads to
\begin{align*}
\hat{\beta}_m  - \tilde\beta_m
= \frac{1}{P_0 g(\mu,w)} & \times
(P_n-P_0) \sum_{k=0}^{K} \biggl[
 \int_{k}^{k+1}Y(t){w}_{km}(t)
 \Delta_{km}(t)
 dM^{(m)}_0(t) \\
& + 
\phi_{km \tilde \alpha} 
 \, \xi^{(3)}_{km \tilde \alpha}
-
 \sum_{j=0}^{m-1} \phi^{\star}_{\tilde\beta_j} 
    \int_{k}^{k+1}  \gamma_{kmj}(t) dt 
\biggr]
+  o_P(n^{-1/2}).
\end{align*}

\clearpage
\section{Additional simulation details and results}\label{app:sim}

\subsection{Model specification}
The parametric exposure model allow for linear effects of covariates $L_1,L_2$, and $Q$. The nonparametric exposure and nuisance models used an ensemble model approach based on the R package SuperLearner. We included in the ensemble a parametric model and lightGBM. The lightGBM approach was incorporated with four different specifications: by combining 50 or 200 rounds with a learning rate of 0.01 or 0.1 and depth of 3.

\subsection{Additional results}

\begin{table}[ht]
\caption{Power for $\beta_0,\beta_1,\beta_2$. Based on proportion of simulation replicates where the lower bound of the 95\% confidence interval was greater than zero.}
\centering
\begin{tabular}{rlrrrrrr}
  \toprule
Scenario & Model & \multicolumn{3}{c}{$n=2000$} & \multicolumn{3}{c}{$n=5000$}\\
&& $\beta_0$ & $\beta_1$ & $\beta_2$ & $\beta_0$ & $\beta_1$ & $\beta_2$ \\ 
  \midrule
1 & Parametric & 0.98 & 0.44 & 0.16 & 1.00 & 0.81 & 0.37 \\ 
   & Nonparametric & 0.99 & 0.39 & 0.18 & 1.00 & 0.82 & 0.36 \\ 
   & Nonparametric Robust & 0.93 & 0.38 & 0.19 & 1.00 & 0.86 & 0.40 \\ 
  2 & Parametric & 0.86 & 0.29 & 0.13 & 1.00 & 0.66 & 0.18 \\ 
   & Nonparametric & 0.85 & 0.34 & 0.14 & 1.00 & 0.67 & 0.22 \\ 
   & Nonparametric Robust & 0.78 & 0.30 & 0.15 & 1.00 & 0.74 & 0.26 \\ 
   \bottomrule
\end{tabular}
\end{table}

\begin{table}[ht]
\caption{Type I error rate for $\beta_3,\beta_4$. Determined by the proportion of simulation replicates where the 95\% confidence interval upper bound was less than zero (`-') or the lower bound was greater than zero (`+'). The sum of `-' and `+' columns give the type I error rate.}
\centering
\begin{tabular}{rlrrrrrrrrrr}
  \toprule
Scenario & Model & \multicolumn{4}{c}{$n=2000$} & \multicolumn{4}{c}{$n=5000$}\\
&& \multicolumn{2}{c}{$\beta_3$} & \multicolumn{2}{c}{$\beta_4$} & \multicolumn{2}{c}{$\beta_3$} & \multicolumn{2}{c}{$\beta_4$} \\ 
&&-&+&-&+&-&+&-&+\\
  \midrule
1 & Parametric & 0.04 & 0.01 & 0.02 & 0.03 & 0.04 & 0.03 & 0.04 & 0.00 \\ 
   & Nonparametric & 0.03 & 0.01 & 0.02 & 0.02 & 0.05 & 0.01 & 0.05 & 0.02 \\ 
   & Nonparametric Robust & 0.07 & 0.06 & 0.05 & 0.03 & 0.10 & 0.05 & 0.04 & 0.03 \\ 
  2 & Parametric & 0.03 & 0.04 & 0.01 & 0.05 & 0.02 & 0.01 & 0.01 & 0.02 \\ 
   & Nonparametric & 0.04 & 0.04 & 0.04 & 0.07 & 0.00 & 0.01 & 0.01 & 0.03 \\ 
   & Nonparametric Robust & 0.06 & 0.03 & 0.05 & 0.08 & 0.02 & 0.02 & 0.03 & 0.03 \\ 
   \bottomrule
\end{tabular}
\end{table}

\section{Additional data analysis details and results}\label{app:data}

For our data analysis we consider 10,438,899 person-months of data representing a longitudinal cohort of $n=299,661$ Medicare Fee for Service (FFS) beneficiaries living in the Northeast region of the United States (as defined by the US Census Bureau, which includes states: Connecticut, Maine, Massachusetts, New Hampshire, New Jersey, New York, Pennsylvania, Rhode Island, and Vermont). This dynamic cohort is a 10\% random sample of all Medicare beneficiaries who turned 65 between January 1$^{st}$, 2000 and December 31$^{st}$, 2010. We follow each individual until death, censoring by leaving Medicare FFS or moving outside the Northeast region of the US, or administrative censoring after 36 months follow-up (Table \ref{tab:cohort_char}); we use the first 6 months of follow-up as a baseline such that $t=0$ at 65.5 years of age for each individual. 

The data structure for each individual in the study cohort consists of five parts: (1) CVD hospitalizations, with day of admission and diagnosis codes for each hospitalization; (2) day of death or censoring, if it occurred during the follow-up period; (3) individual-level covariates that include date of birth, race, sex, ZIP code of residence, and Medicaid eligibility (a proxy for low-income); (4) neighborhood characteristics based on US Census and annual American Community Survey data; (5) exposures to air pollution in the neighborhood (Figure \ref{fig:pm25_map}). Neighborhood characteristics and exposures are resolved at postal ZIP codes. 

\begin{table}[!ht]
\caption{Cohort characteristics by year of entry into Medicare FFS at age 65. Every individual is followed for a baseline period of 6 months and then during a 30 month study period. Individuals who were censored during the baseline period are not included in the study.}\label{tab:cohort_char}
\centering
\scriptsize
\begin{tabular}{rrrrrrrrrrrr}
  \toprule[1pt]
 & 2000 & 2001 & 2002 & 2003 & 2004 & 2005 & 2006 & 2007 & 2008 & 2009 & 2010 \\ 
  \midrule
N & 22942 & 25207 & 25563 & 26884 & 27474 & 26761 & 27551 & 27572 & 30523 & 30311 & 28873 \\ 
\multicolumn{1}{l}{\textbf{Sex}}\\
Male & 10796 & 11812 & 12001 & 12678 & 12922 & 12707 & 13003 & 13017 & 14571 & 14539 & 13842 \\ 
Female &   12146 & 13395 & 13562 & 14206 & 14552 & 14054 & 14548 & 14555 & 15952 & 15772 & 15031 \\ 
\addlinespace
\multicolumn{1}{l}{\textbf{Race}}\\
White & 19884 & 21857 & 22254 & 23237 & 23822 & 22923 & 23635 & 23615 & 26419 & 26287 & 24539 \\ 
Black & 2006 & 2198 & 2092 & 2235 & 2292 & 2287 & 2306 & 2282 & 2304 & 2278 & 2372 \\ 
Asian &  260 & 290 & 317 & 390 & 342 & 411 & 490 & 504 & 509 & 487 & 561 \\ 
Hispanic &  252 & 279 & 267 & 333 & 341 & 389 & 346 & 355 & 397 & 389 & 432 \\  
Other/Unknown & 540 & 583 & 633 & 689 & 677 & 751 & 774 & 816 & 894 & 870 & 969 \\ 
\addlinespace
\multicolumn{2}{l}{\textbf{Medicaid eligible}}\\
No & 20751 & 22867 & 23201 & 24353 & 25003 & 24296 & 25022 & 25202 & 28082 & 27839 & 26250 \\ 
Yes & 2191 & 2340 & 2362 & 2531 & 2471 & 2465 & 2529 & 2370 & 2441 & 2472 & 2623 \\ 
\addlinespace
\multicolumn{6}{l}{\textbf{Terminal events or censoring during study period}}\\
Died age $<68$ & 1089 & 1138 & 1087 & 1058 & 1067 & 1008 & 991 & 963 & 1114 & 1058 & 992\\
Censored age $<68$ & 1995 & 2086 & 2295 & 3147 & 3858 & 4182 & 4475 & 4398 & 3991 & 3790 & 3747\\
   \bottomrule[1pt]
\end{tabular}
\end{table}

\begin{figure}[!ht]
    \centering
    \subcaptionbox{Distribution of annual average PM$_{2.5}$ concentration ($\mu g/m^3)$ across the Northeast region of the United States. The color scale indicates exposure concentration with three years shown to depict the time-varying nature of exposure.\label{fig:pm25_map}}{\includegraphics[width=\textwidth]{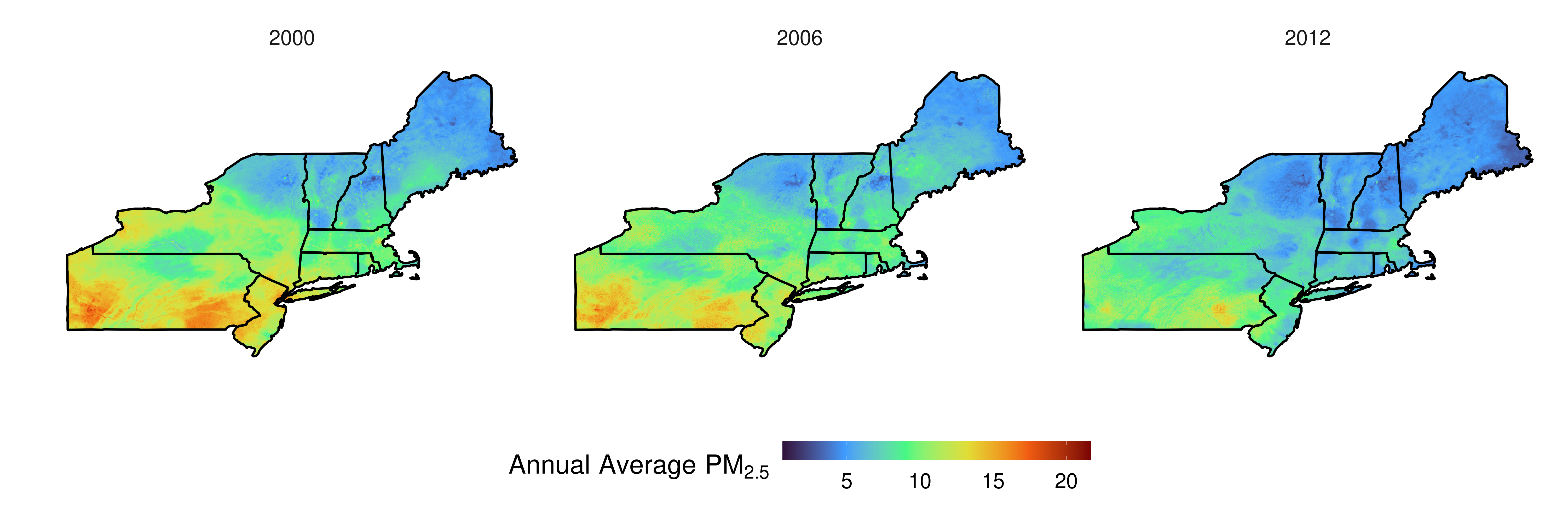}}
    \subcaptionbox{Smoothed trends of monthly average PM$_{2.5}$ concentration ($\mu g/m^3)$ among sampled individuals, separated by state (color and line pattern).\label{fig:pm_trend}}{\includegraphics[width=0.8\textwidth]{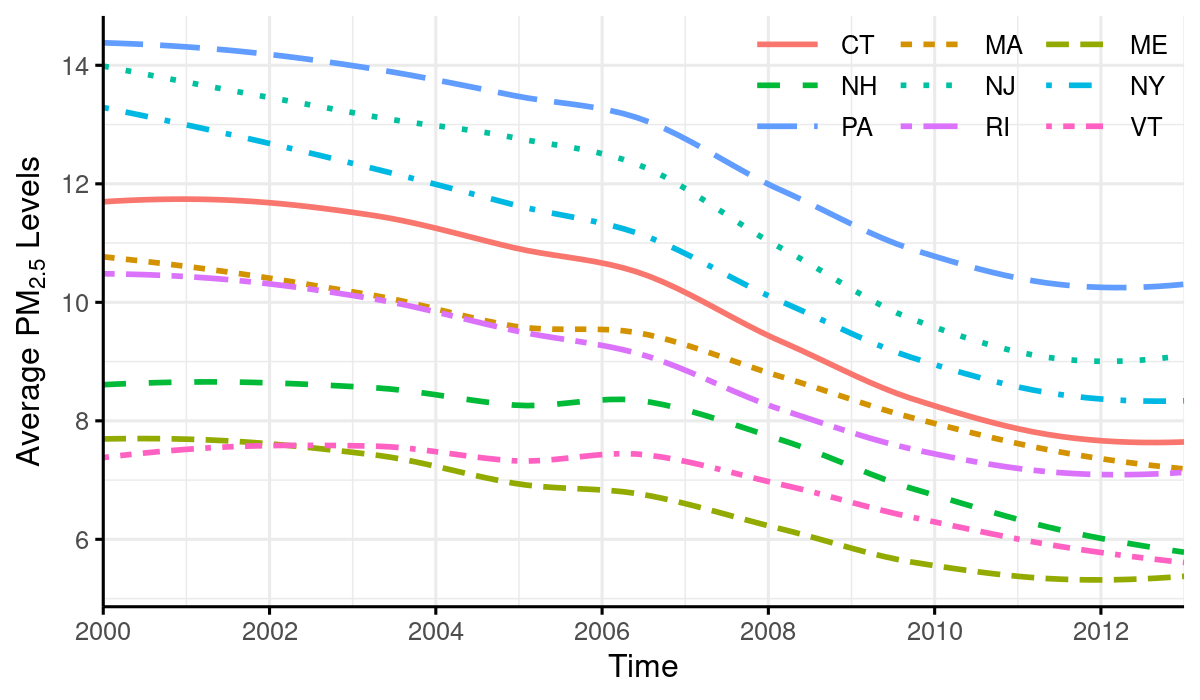}}
    \caption{PM$_{2.5}$ trends in Northeast region of the United States.}
\end{figure}

To ensure a proper baseline confounding adjustment, we consider the baseline period as the 6 months after an individual turns 65. Within this time-frame we record baseline exposure and covariates, including past CVD hospitalizations.  It is important to note that time is relative to each individual, as the nature of the study design allows people to enter at different calendar times. To account for the dynamic cohort, calendar time (e.g., year and month) is a covariate in the study.

\subsection{CVD Hospitalizations}
Hospitalization admissions are recorded in Medicare Part A MedPAR files. We identify hospitalizations as CVD-related based on the primary medical diagnosis billing code that follows the International Classification of Diseases (ICD, including the revisions ICD-9 for the years 2000-2015 and ICD-10 for years 2015-2016). We consider three diagnostic classifications for CVD hospitalization: atrial fibrillation (ICD-9: 427.3; ICD-10: I48), cardiac arrest (ICD-9: 427.5; ICD-10: I46), and acute myocardial infarction (ICD-9: 410; ICD-10: I21). If there are multiple hospitalization records within a 2-day period (for example, often an individual is discharged from one hospital and admitted to another), we only consider the first as a hospitalization event.

\subsection{Exposure and Covariates}
We obtained daily exposure predictions for PM$_{2.5}$ on a 1 km$^2$ grid from a well-validated ensemble learning model that achieved a cross-validated $R^2$ of 0.89 \citep{Di2019AnResolution}. We calculated monthly exposures at each ZIP code by averaging daily exposure from all grid centroids within the ZIP code boundaries. We assigned the monthly PM$_{2.5}$ measurements to each individual based on their ZIP code of residence from MBSF. Figure \ref{fig:pm25_map} shows the spatial distribution of exposure while \ref{fig:pm_trend} shows the PM$_{2.5}$ trend by state among Medicare beneficiaries included in our study. We see higher exposures on average across New Jersey, New York, and Pennsylvania compared to other states. In 2000, only Maine, New Hampshire and Vermont had an average exposure below 9 $\mu g/m^2$, while at the end of 2012 all states except New Jersey and Pennsylvania have average PM$_{2.5}$ concentrations below this level.

The Medicare Master Beneficiary Summary File (MBSF) contains individual information on date of birth, sex, race (White, Black, Hispanic, Asian/Pacific Islander, Native American, Other), an indicator of Medicaid eligibility, time enrolled in private Medicare Advantage supplemental insurance plans, ZIP code of residence, and date of death. Using the ZIP code of residence, we link neighborhood-level demographic and socioeconomic variables derived from the 2000 US Census and 2011-2016 annual American Community Survey (ACS). Area-level demographic variables include population density, percent of individuals who identify as Hispanic, Black, Asian, White, or other race/ethnicity, and percent of individuals who graduated from high school. Socioeconomic variables include median home price, median home income, percentage of owner-occupied housing, and percent of individuals living below the federal poverty level. We also link county-level average body mass index and percent of smoking population from the Center for Disease Control's Behavior Risk Factor Surveillance System surveys from 2000-2012. Area-level covariates missing during a particular year in a ZIP code are imputed by linear interpolation. Individuals who changed ZIP codes during follow-up were excluded from the study as the month of year of the move was unknown. Non-administrative censoring occurred if missing covariates were present.

\subsection{Data limitations}
There are several limitations in the data that are important to mention. First, many of the area-level confounders are updated annually while the time-varying PM$_{2.5}$ used in the analysis are updated monthly. However, area-level confounders are relatively steady year-over-year (average correlation = 0.936) and machine learning approaches for the exposure model ($\mu_{km}$) can account for differences within years using indicators for each month. Second, the monthly ZIP code average exposure may be subject to measurement error with respect to each individual living in that ZIP code. Previous work has assessed the sensitivity of estimated health effects due to the potential misalignment of spatially aggregated exposure and indicated a bias toward the null \citep{Kioumourtzoglou2014ExposureStudies,Wei2022TheMortality}. Future work to incorporate regression calibration to accommodate exposure measurement error could reduce the potential bias in our models \citep{Josey2023EstimatingMortality}. Third, we note that subjects are not truly at risk for rehospitalization while currently in the hospital (i.e., between admission and discharge); our choice of at risk process does not account for this ``gap'' in being at risk. However, this is not considered a significant limitation of our analyses. In particular, the median time from hospital admission to discharge is 3 days with 25\% hospitalizations lasting longer than 6 days and 5\% longer than 13 days; given that there are only 0.75 primary-cause CVD hospitalizations for every 1,000 person-months in our data and only 1,051 individuals (0.35\%) with 2 or more hospitalizations during the 30 month follow-up, the amount of time not at risk is a tiny fraction of available follow-up. Certainly, additional considerations for the at risk process, $Y(t)$, may be required for data sets with much larger event rates and long periods when individuals cannot record any events.

\subsection{Model Specification}
The parametric estimation approach for the exposure model $\mu_{km}(t)$ included linear effects of all neighborhood covariates measured during period $k-m$, linear effects for other exposures (monthly average NO$_2$, O$_3$, temperature, humidity) from periods $k-m$ through $k-m-6$ and linear effects for monthly average PM$_{2.5}$ from periods $k-m-1$ through $k-m-6$, and intercepts for every month and year combination. The model also parameterized a smooth effect of year (natural spline with $df=3$) for each state and smooth effect of year (natural spline with $df=3$) interacted smooth effect of ZIP code centroid (natural spline with $df=5$ for each latitude and longitude).

The nonparametric exposure and nuisance models used the same ensemble model approach from the simulation, based on the R package SuperLearner. We included in the ensemble a parametric model with linear effects of each covariate and lightGBM. The lightGBM approach was incorporated with four different specifications: by combining 50 or 200 rounds with a learning rate of 0.01 or 0.1 and depth of 3. All covariates used in the parametric estimator were also included in the robust estimator. In the nuisance function, we also included individual-level covariates race, sex, Medicaid eligibility, and number of prior CVD hospitalizations.

\clearpage
\subsection{Additional results}
\begin{table}[!ht]
\caption{Individual lag estimates (bootstrap SD) of death or CVD hospitalization rate per 100,000 Medicare beneficiaries resulting from a 10$\mu g/m^3$ increase in PM$_{2.5}$.}
\centering
\small
\begin{tabular}{rrrrr}
  \toprule
Lag & \multicolumn{1}{c}{$\alpha$} & $\beta$ (Parametric) & $\beta$ (Nonparametric) & $\beta$ (Nonparametric Robust) \\ 
  \midrule
  0 & 33.61 (15.66) & 12.51 (10.70) & 14.65 (6.66) & 6.71 (7.50) \\ 
    1 & -3.44 (15.65) & 3.53 (12.50) & 10.77 (6.46) & 9.43 (7.54) \\ 
    2 & 0.27 (16.12) & -6.64 (12.41) & 4.76 (5.34) & 0.90 (6.71) \\ 
    3 & 11.48 (16.02) & 9.19 (13.19) & 0.89 (6.79) & -2.09 (8.51) \\ 
    4 & 26.15 (16.55) & -1.97 (12.42) & -1.31 (6.91) & -7.19 (8.10) \\ 
    5 & -34.54 (17.09) & 14.57 (13.27) & 6.78 (7.53) & 8.12 (7.39) \\ 
    6 & 18.46 (15.44) & -5.68 (10.99) & 8.68 (7.49) & 11.4 (7.52) \\ 
   \bottomrule
\end{tabular}
\end{table}

\begin{table}[!ht]
\caption{Cumulative effect ($\sum_{m=0}^l\beta_m$) estimates (bootstrap SD) of death or CVD hospitalization rate per 100,000 Medicare beneficiaries resulting from a 10$\mu g/m^3$ increase in PM$_{2.5}$.}
\centering
\small
\begin{tabular}{rrrrr}
  \toprule
Lag & \multicolumn{1}{c}{$\alpha$} & $\beta$ (Parametric) & $\beta$ (Nonparametric) & $\beta$ (Nonparametric Robust) \\ 
  \midrule
  0 & 33.61 (15.66) & 12.51 (10.70) & 14.65 (6.66) & 6.71 (7.50) \\ 
    1 & 30.17 (14.78) & 16.05 (11.63) & 25.42 (8.42) & 16.14 (9.03) \\ 
    2 & 30.44 (17.26) & 9.41 (12.71) & 30.18 (9.40) & 17.04 (10.45) \\ 
    3 & 41.92 (14.10) & 18.60 (11.81) & 31.06 (10.23) & 14.95 (11.82) \\ 
    4 & 68.07 (15.31) & 16.63 (14.08) & 29.76 (11.68) & 7.75 (14.40) \\ 
    5 & 33.53 (18.80) & 31.20 (15.05) & 36.54 (11.14) & 15.87 (12.44) \\ 
    6 & 52.00 (19.33) & 25.52 (13.83) & 45.22 (12.18) & 27.27 (12.55) \\ 
   \bottomrule
\end{tabular}
\end{table}

\end{document}